\documentclass[a4paper,11pt]{article}
\pdfoutput=1
\usepackage[utf8]{inputenc}
\usepackage[english]{babel}
\usepackage[T1]{fontenc}
\usepackage{amsmath,amssymb}
\usepackage{braket}
\usepackage{hyperref}
\usepackage{multirow}
\usepackage{tikz}
\usepackage{lipsum}
\usepackage{dcolumn}
\usepackage{color}
\usepackage{tikz}
\usetikzlibrary{quantikz2}
\usepackage{algorithm}
\usepackage{algpseudocode}
\usepackage{setspace}
\usepackage{multirow}
\usepackage{amsthm}
\theoremstyle{plain}
\newtheorem{lemma}{Lemma}
\usepackage{appendix}
\usepackage{mathtools}
\usepackage{ulem}
\usepackage{enumitem}
\usepackage{graphicx} %
\graphicspath{{figures/}} %
\usepackage{tikz}
\usepackage{pgfplots} %
\usepackage{pgfmath} %
\usetikzlibrary{positioning}%
\usetikzlibrary{arrows}%
\usetikzlibrary{calc}
\usepackage{authblk}

\usepackage{booktabs}
\usepackage[margin=1in]{geometry}
\usepackage{array}
\newcolumntype{L}[1]{>{\raggedright\arraybackslash}p{#1}}

\newcommand{\qdb}[1]{\ket{\text{QDB}^{(#1)}}} %

\definecolor{CeruleanRef}{RGB}{12,127,172}
\hypersetup{colorlinks=true}
\hypersetup{allcolors=CeruleanRef}

\newtheorem{theorem}{Theorem}
\usepackage{amsmath,amssymb}
\usepackage{braket}

\newcommand{\RoundRect}[4]{\draw[rounded corners=5,black!60!white,fill=black!5!white,] (#1,#2) rectangle ++(#3,#4);}

\newcommand{\RoundRectGrid}[2]{%
\RoundRect{0}{0}{#1}{#2}
\RoundRect{#1 + 0.075}{0}{#1}{#2}
\RoundRect{0}{#2 + 0.075}{#1}{#2}
\RoundRect{#1 + 0.075}{#2 + 0.075}{#1}{#2}
}

\begin{document}

\title{Operational Framework for a Quantum Database}

\author[1,2]{Carla Rieger}
\author[2]{Michele Grossi}
\author[3]{Gian Giacomo Guerreschi}
\author[2]{Sofia Vallecorsa}
\author[1]{Martin Werner}

\affil[1]{School of Engineering and Design, Technical University of Munich, Germany}
\affil[2]{European Organization for Nuclear Research (CERN), Geneva 1211, Switzerland}
\affil[3]{Intel Labs, Intel Corporation, 2200 Mission College Blvd, Santa Clara, CA 95054, United States of America}

\maketitle
\begin{abstract}
Databases are an essential component of modern computing infrastructures and allow efficient manipulation of inherently structured data. The structure depends on the type and relationships of the individual data elements and on the access pattern.
Extending the concept of databases to the quantum domain is expected to increase both the storage efficiency and the access parallelism through quantum superposition. In addition, quantum databases may be seen as the result of a prior state preparation ready to be used by quantum algorithms when needed. On the other hand, limiting factors exist and include entanglement creation, the impossibility of perfect copying due to the no-cloning theorem, and the impossibility of coherently erasing a quantum state.

In this work, we introduce quantum databases within the broader context of data structures using classical, or more precisely \textit{cloneable}, and quantum data and indexing. In particular, we are interested in quantum databases' practical implementation and usability,  focusing on the definition of the basic operations needed to create and manipulate data stored in a superposition state. Specifically, we address the case of quantum indexing in combination with cloneable data. For this scenario, we define the operations for database preparation, extension, removal of indices, writing, and read-out of data, as well as index permutation. We present their algorithmic implementation and highlight their advantages and limitations. Finally, we introduce the next steps toward defining the same operations in the more general context of quantum indexing and quantum data.
\end{abstract}

\section{Introduction}
\label{sec:introduction}
\vspace{-3mm}
Managing and analyzing large amounts of data is a fundamental task in the big-data age when most decisions are expected to be data-driven. At the same time, new technologies like quantum computing have been developed to solve computational bottlenecks. It is natural to ask whether the advantages also extend to the database space.
One of the key components for quantum computational advantage, e.g., in Grover's algorithm~\cite{grover1996fast}, is the presence of quantum superposition states that allow for inherent parallelism and interference. Furthermore, quantum systems have access to a state space that is exponentially large with respect to the number of qubits, which permits a compressed representation of data.
Intuitively, to take full advantage of this representation, one needs to process the quantum state without converting it to classical information, e.g., by avoiding measurement.
Prominent algorithms which may operate on quantum databases include pattern recognition~\cite{trugenberger2002quantum, niroula2021quantum}, collision finding~\cite{Brassard_1998}, quantum search of an unstructured database~\cite{grover1996fast}, solving linear systems of equations~\cite{PhysRevLett.103.150502} and many more.\newline
\begin{figure}[h!]
 \centering

\begin{tikzpicture}[scale=3]

\begin{scope}[shift={(6.6,0)}]

\RoundRectGrid{0.85}{0.85}

\node[align=center] at (0.425, 1.9) {\small classical};
\node[align=center] at (1.36, 1.9) {\small quantum};
\node[align=center] at (0.9, 2.06) {\large Type of data};

\node[align=center, rotate=90] at (-0.31, 0.9) {\large Type of indexing};
\node[align=center, rotate=90] at (-0.14, 1.32) {\small classical};
\node[align=center, rotate=90] at (-0.14, 0.415) {\small quantum};

\node[align=center] at (0.425, 1.35) {\Huge CC};
\node[align=center] at (1.36, 1.32) {\Huge CQ};
\node[align=center] at (0.425, 0.415) {\Huge QC};
\node[align=center] at (1.36, 0.415) {\Huge QQ};

\end{scope}
\end{tikzpicture}
\caption{The different types of database structured with respect to classical and quantum states as index and data ($C$: classical, $Q$: quantum). Operations defined for each possible combination of type of data and indexing for a database differ vastly. We discuss the term \textit{classical} data in a more general context of \textit{cloneable} data.}
\label{fig:quadrants}
\end{figure}
\newline 
Generally, we can characterize a database formalized by the data type stored and how elements are indexed. Thus, we can distinguish four distinct settings as shown in Figure~\ref{fig:quadrants}. This work focuses on designing database manipulation operations for the case of quantum indexing, mainly with respect to classical data encoded in a set of qubits. The community has discussed this data structure also within the context of a quantum random access memory (QRAM)~\cite{giovannetti2008quantum, giovannetti2008architectures} that forms dedicated hardware structures for a coherent read and write operation implementing the oracle call described below in~\eqref{eqn:qram_call}. Thereby, the QRAM query performance is reduced with respect to its classical counterpart as fewer logic gates need to be activated. In this model, the resulting state is a register state indexed by the state~$\ket{j}$ that points to the data element~$\ket{d_j}$. A QRAM memory call returns a superposition of \textit{data} states upon providing a superposition of indices~\cite{giovannetti2008quantum}:
\begin{align}
     \sum_{j} \psi_j |j\rangle \xmapsto{\text{QRAM call}} \sum_{j} \psi_j |j\rangle |d_{j}\rangle \, .
     \label{eqn:qram_call}
\end{align}
In~\eqref{eqn:qram_call}, the amplitudes are normalized, i.e.,~$\sum_j |\psi_j|^2=1$. Within the context of QRAM, a dedicated hardware structure for implementing a QRAM call described in~\eqref{eqn:qram_call} is discussed. In this work, on the contrary, we focus on an algorithmic and, therefore, a software-based framework to be executed directly in the quantum processing unit of a quantum device. We note that if one has access to the dedicated hardware structure of a QRAM, one may use it to prepare a quantum database state. Furthermore, a QRAM may be used to convert classical to quantum indices and thus convert a CC/CQ database (as in Figure~\ref{fig:quadrants}) to a QC/QQ database. On the other hand, a dedicated QRAM does not have to be available to use our software-centric framework. \newline
Recently, the definition of quantum data centers (QDC)~\cite{liu2023data,liu2023quantum} as a QRAM combined with a quantum network has been proposed. In this proposal~\cite{liu2023data}, the applications range from usage as a $T$-gate resource for multiparty private quantum communication to distributed sensing through data compression. 
Thus, superposition states as returned by a QRAM call described in~\eqref{eqn:qram_call} are relevant in many ways. In the context of quantum databases, recent works aim to define database manipulation operations acting on a quantum superposition state~\cite{younes2013database, liu2023data,hai2024quantum} hinting at the differences with respect to classical and quantum data. 

Several fundamental quantum mechanical phenomena limit the operations that can be done within the framework of a quantum database, including the \textit{no-cloning}~\cite{wootters1982single}, and the \textit{no-deletion theorem}~\cite{pati2006no, kumar2000impossibility}. Other works~\cite{younes2013database, liu2023data,hai2024quantum} consider database operations such as \textit{select}, \textit{extend} and \textit{delete} in a more simplified manner or do not provide an implementation. An example is given by the Flip-Flop QRAM~\cite{park2019circuit}, which demonstrates a specific gate-based implementation for \textit{writing} classical data in a superposition state.

In this work, we introduce the specific set of operations that we believe are essential to operating a database in practice and that, thus, constitute the basis of our extension from classical to quantum databases. Their definition is presented in Table~\ref{table:mr}, and their implementation will be discussed for specific cases. In the general case of quantum indexing, we define a quantum database as being the state:
\begin{align}
    \qdb{k} = \sum_{j=0}^{k-1} \psi_j \ket{j}_I \ket{d_j}_D \in \mathcal{H}_I \otimes \mathcal{H}_D \,,
    \label{eqn:general-qdb}
\end{align}
with $\sum_j |\psi_j|^2 =1$. In the following, we often consider the restricted case of balanced database entries, namely in~\eqref{eqn:general-qdb} we have~$\psi_j=1/\sqrt{k}$ for~$\forall j$. Uniform superposition states are of relevance for algorithms, including, e.g., Grover's algorithm for database search~\cite{grover1996fast} and the Quantum Byzantine Agreement~\cite{fitzi2001quantum}. Hence, we operate on a \textit{known} amplitude distribution to demonstrate our protocols. For certain algorithms, we exclude the amplitude~$\psi_0$ corresponding to the index~$\ket{0}_I\ket{0}_D$ from being balanced and use this state as a \textit{probability reservoir} that can be used to add new data elements or be replenished when data are removed. The superposition state in~\eqref{eqn:general-qdb} consists of a separate index~($I$) and a data register~($D$). The index states are orthonormal and belong to the Hilbert space~$\mathcal{H}_I$ of dimension at least~$k$. The data states are part of the Hilbert space~$\mathcal{H}_D$ of dimension~$m$ composed of~$\Tilde{m}$ qubits with~$m = 2^{\Tilde{m}}$. For simplicity, we assume that both spaces are formed by qubits. The superposition state in~\eqref{eqn:general-qdb} is a linear combination of~$k$ elements. For the data elements, we used the shorthand notation~$\ket{d_j} = \ket{d_{\alpha_j}}$ indicating that the~$j$-th data is chosen from a set of states~$\{\ket{d_{\alpha}} \}_{\alpha=0,1,\dots}$ in~$\mathcal{H}_D$. The properties of such a set of states determine if we consider the data as quantum or classical. Specifically, classical data correspond to orthogonal states~$\langle{d_{\alpha}}\ket{d_{\beta}}=\delta_{\alpha \beta}$ that can be mapped to the computational basis by a known transformation~$U_D$. Notice that the data elements in~$\qdb{k}$ do not have to be chosen uniquely, meaning that it is not necessary that~$\alpha_j \neq \alpha_i$ for~$j \neq i$ in~\eqref{eqn:general-qdb}. 
The number of qubits is chosen to be the smallest integer that fulfills~$\Tilde{k} = \lceil \log_2(k) \rceil$. By $\lceil x \rceil$ we denote the ceiling function that returns the smallest integer greater than or equal to~$x$. We consider the different scenarios for orthogonal~$|d_j \rangle_D$ and non-orthogonal data entry states~$|\Tilde{d}_j \rangle_D$, discuss the former situation in detail and then present an outlook for the operations on the latter.

The rest of this article is structured as follows. We start by introducing the case of classical indexing in Section~\ref{sec:cc} and contrast it with the case of quantum indexing in Section~\ref{sec:qc}. For the case of quantum indexing and classical data states, we formally define the operations: 
\begin{center}
    \text{\textit{prepare}, \textit{extend}, \textit{remove},} \\ \text{\textit{write}, \textit{read-out} and \textit{permute}} \nonumber
\end{center}
and discuss their algorithmic implementation as summarized in Table~\ref{table:mr}. We then focus on the quantum database extension algorithm presented in Section~\ref{sec:extension_chapter}. Following that, we summarize implications on the defined operations in the case of quantum data in Section~\ref{sec:qq} as some of the operations are equivalent to the classical data case and conclude this work in Section~\ref{sec:conclusion}.

\section{Data structured by classical indices~(CC and CQ)}

\label{sec:cc}
In the classical computer science domain, database models encompass structured frameworks for efficiently organizing, storing, and managing data, including hierarchical, network, relational, and object-oriented models. Briefly, the hierarchical model arranges data in a tree-like structure, while the network model uses a graph structure to handle complex many-to-many relationships. The relational model stores data in normalized tables, which are accessed via structured query languages (see, e.g.,~\cite{date1989guide}), offering robust querying and indexing capabilities. Database models emphasize the importance of individual data relationships in maintaining the integrity of data and facilitating complex queries. Due to the maturity of this field, we focus on a form of relational database structures that organize and access the data by an integer index. The classical setup is well known and commonly used, and thus, we do not go further into detail but highlight the relevant quantum extension from now on.\newline
\noindent 
If we consider quantum data accessed by classical indexing, we think of the following setup:
\begin{align}
    \bigotimes_{i=0}^p \ket{{d}_i}_D \, \in \left( \mathcal{H}_D \right)^{\otimes p}\,,
    \label{eqn:linear_cq} 
\end{align}
where the data element~$\ket{d_i}_D$ associated with the index~$i$ resides in a dedicated register. In~\eqref{eqn:linear_cq} there are~$p$ individual registers consisting each of~$\Tilde{m}$ qubits such that~$\dim (\mathcal{H}_D)=2^{\Tilde{m}}$. The data may be drawn from any orthogonal or non-orthogonal set of states. For this approach, operations that act on a system as given in~\eqref{eqn:linear_cq} are well defined and can be formulated as the individual registers are in product form. An extension corresponds to adding a new quantum state to the system. One can swap individual registers without creating entanglement. Due to the product form of~\eqref{eqn:linear_cq}, individual registers can be traced out (and thus neglected and removed) without affecting the remaining states. This setup is based on a linear number of quantum resources. Due to this \textit{uncompressed} way of storing the data, the framework of~\eqref{eqn:linear_cq} is inefficient in the number of resources corresponding here to the number of qubits. Thus, in the following section, we consider the case of \textit{quantum indexing}. 

\begin{center}
\vspace{-4mm}  
\begin{table}[h!]
\caption{Set of defined quantum database~(QDB) operations that we consider to be essential QDB operations. We summarize their respective actions on the QDB state and draw an analogy to operations in the classical database~(DB) scenario to illustrate possible correspondences.}
\centering
\begin{tabular}{@{}p{0.16\textwidth}p{0.5\textwidth}@{\hspace{1cm}}p{0.1\textwidth}*{4}{l}@{}}
\toprule
QDB operation & QDB action & Name & DB Analog \\
\midrule
Prepare & Prepare empty QDB with $k$ index elements. & $P_{(k)}$ & CREATE \\ 
Extend & Extend $\qdb{k}$ by $l$ new index elements. & $E_{(l)}$ & RESERVE \\ 
Remove & Remove index $f$ with data element from $\qdb{k}$. & $R_{(f)}$ & DELETE \\ 
Write & Write data element with specific index $f$. &  $W_{(f)}$ & INSERT \\ 
Read-out & Read-out the data element indexed by $f$. & $G_{(f)}$ & SELECT \\ 
Permute & Permute the index elements based on $\pi$. & $\text{P}_{\pi}$ & - \\ 

\bottomrule
\end{tabular}
\label{table:mr}
\vspace{-9mm}
\end{table}
\end{center}

\section{Data structured by quantum indices}

When the index is itself a quantum state, we allow for superpositions of indices, each correlated to a data state. In general, we consider balanced superposition states as in~\eqref{eqn:general-qdb}. For some of the algorithms introduced below, we exclude the~$\ket{0}_I \,$-index state from being balanced as it forms a reference and is used as a \textit{probability reservoir}. This term comes from the fact that measuring the index register collapses the database to a single index-data pair with a probability given by the squared absolute value of the amplitude $\psi_j$.

\subsection{Classical or \textit{cloneable} data (QC)}
\label{sec:qc}
In this section, we consider classical data states indexed by quantum indices. As a reminder, by~$\ket{d_j}$ we denote a state non-uniquely drawn from an orthogonal set of states~$\{ d_{\alpha} \}_{\alpha}\in \mathcal{H}_D$. We assume that \textit{classical}, or more precisely \textit{cloneable}, data is defined as being drawn from the aforementioned set of orthogonal states, and it is given that we know the unitary transformation~$U_D$ that maps each element in this set to the computational basis.

\subsubsection{Preparation of an empty QDB}
\label{sec:prepare}
To prepare a database with balanced, uniformly distributed amplitudes, quantum indexing, and an empty data register, we apply the following operation:
\begin{align}
    & \mathcal{H}_I \otimes \mathcal{H}_D \rightarrow \mathcal{H}_I \otimes \mathcal{H}_D \nonumber \\
    &|0 \rangle_{I \otimes D}  \xmapsto{P_{(k)}} \label{eqn:prepare}  \frac{1}{\sqrt{k}} \sum_{j = 0}^{k-1} |j \rangle_I |0 \rangle_D =: |\text{QDB}^{(k)}_{\text{empty}}\rangle \, . 
\end{align}
In~\eqref{eqn:prepare},~$P_{(k)}$ is particularly simple when~$\log_2(k) \in \mathbb{N}_{>0}$.
This case corresponds to the Walsh-Hadamard transform applied only to the index register~$I$ with a circuit depth of~$\mathcal{O}(1)$. The indices remain in a product state with respect to the empty data register~$D$. Thus, the operation is given as follows:
\begin{equation}
    P_{(k)} = H^{\otimes \log_2(k)} \otimes \mathbb{I}_D \, .
    \label{eqn:database_preparation}
\end{equation}
By creating the state $\qdb{k}_{\text{empty}}$, we have prepared an empty database with~$k$ indices in superposition. All indices are correlated to the empty data string given by~$|0\rangle_D$.
The general case of the preparation operation for a balanced superposition state with~$k$ indices and $\log_2(k) \notin \mathbb{N}_{>0}$ is presented in Appendix~\ref{sec:prepare_general} together with the case including the probability reservoir. The former case without a probability reservoir has been recently discussed in~\cite{shukla2024efficient}.

\subsubsection{Extend the QDB by new indices correlated to an empty data string}
\label{sec:extension_chapter}

\begin{figure}
    \centering
    \includegraphics[width=0.4\textwidth]{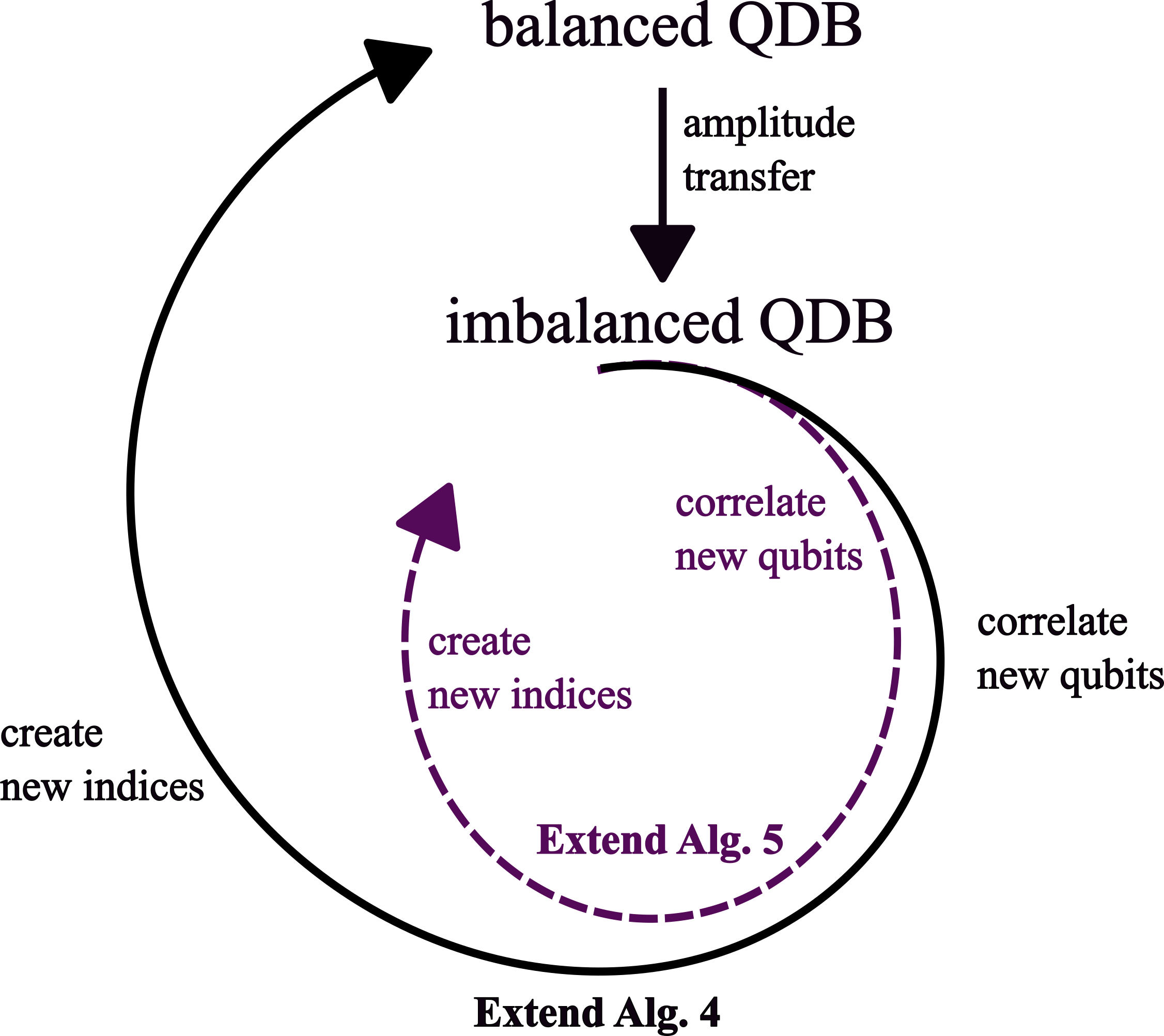}
    \caption{Schematic of the quantum database extension procedure consisting of adding new qubits and creating new indices that are correlated to the empty data string. In contrast to Algorithm~\ref{alg:combined_transfer}, Algorithm~\ref{alg:imbalanced_extend} does not include the \textit{amplitude transfer} step and uses a specifically imbalanced database with the zero-index state acting as a probability reservoir. Both protocols are described in detail in Section~\ref{sec:extend_database}. }
    \label{fig:schematic-extend}
\end{figure}
The operation of extending the quantum database corresponds to expanding the number of indices in the superposition and, crucially, correlating them to the empty data string. This operation seems straightforward, but we will see that the possibility of having entanglement between index and data registers is very consequential. In general, we may correlate a number of new ancilla qubits to the index basis for the index extension procedure and aim to add as many new indices as possible. If the database initially holds a superposition of~$k$ indices, adding a qubit to the index basis doubles the maximal number of possible indices to~$2k$. Next to the demand that the new indices be correlated to the empty data string, the database should remain in a balanced superposition with uniformly distributed amplitudes. The extension operation for adding~$l>0$ new elements is described by:
\begin{align}
    &\mathcal{H}_I \otimes \mathcal{H}_D  \longrightarrow \mathcal{H}_{I'}  \otimes \mathcal{H}_D := \left(\mathcal{H}_{0} \otimes \mathcal{H}_{I} \right)  \otimes \mathcal{H}_D \nonumber \\
    &\qdb{k} = \frac{1}{\sqrt{k}} \sum_{j=0}^{k-1}  \ket{j}_I \ket{d_j}_D \label{eq:extend} \\ 
    &\xmapsto{E_{(l)}} \qdb{k+l} = \frac{1}{\sqrt{k+l}} \sum_{j=0}^{k+l-1}  \ket{j}_{I'} \ket{d_j}_D \, , \nonumber 
\end{align}
with 
\begin{align}
    \nonumber |d_j\rangle_D = | 0 \rangle_D \text{ for } \forall j > (k-1) \text{ and } j = 0 \, . 
\end{align}
The transformation in~\eqref{eq:extend} is non-trivial to define as there is no general unitary map~$E_{(l)}$ fulfilling it. A proof of this fact is presented in Appendix~\ref{sec:extend_problem}. Thus, the unitary operation~$E_{(l)}$ depends on the specific database and is not generally applicable. 

In the simplest case of~\eqref{eq:extend} that still requires adding an extra index register~$\mathcal{H}_0$, we add a single qubit to the index basis and consider~$\log_2(k) \in \mathbb{N}_{>0}$, that forms an initial database with all available indices already in use. Thus, we would have $\dim(\mathcal{H}_0) = 2$ and $\mathcal{H}_{I'} = \mathcal{H}_0 \otimes \mathcal{H}_{I}$. In this case, we can add~$l$ new index elements with~$0< l \leq k$. The data string corresponding to the newly added indices is empty.
\noindent
We present two different protocols for the extension of the quantum database:
\begin{enumerate}
    \item \textit{Insert a QDB containing data, transfer amplitude to zero-index, and create new indices.} 
    \item \textit{Initialize an empty database with an adjusted amplitude distribution, write data in the data registers, and create new indices afterward.}
\end{enumerate}
Regarding the first version, the quantum index extension protocol consists of two steps: an amplitude transfer operation to ensure the amplitudes remain balanced and an index creation step. The amplitude transfer operation is built on the generalized amplitude amplification algorithm with modified step sizes as presented by Brassard~\textit{et al.}~\cite{brassard2002quantum}. The modified amplitude amplification algorithm generally allows us to find a good solution with certainty~\cite{brassard2002quantum} and can be adapted to set a known amplitude to any other desired value. As this modification allows the transfer of small amplitudes onto a specific subspace, we may use it to re-balance the amplitudes of the superposition. Specifically, we transfer amplitude to the zero index state that we use as a \textit{probability reservoir}. As this state is correlated to the empty data state, we create new indices from it based on the generalized prepare method described in Section~\ref{sec:prepare_general} after correlating an ancilla register. The whole procedure is costly as it involves knowing and applying the unitary~$\mathcal{U}_{\qdb{k}}$ (and its inverse~$\mathcal{U}_{\qdb{k}}^{\dagger}$) that creates the whole database (including the data elements) during the amplitude transfer step. There is a need for a less demanding approach that is not based on amplitude amplification.

The second protocol starts from an unbalanced database and can add new elements only until the probability reservoir is fully depleted. Thus, it is important to know an estimate or at least an upper bound on the number of data elements that one desires to store in the database~(denoted here by~$k+l$)
This version initializes a database with a \textit{higher} amplitude on the zero index state. Next, the data registers are filled, an ancilla register is added, and new index states are created that are correlated to the empty data state. 
In the Appendix~\ref{sec:extend_database}, we go into the details for both variants of the extension protocol. The second approach appears more practical, and we observe that it uses a controlled \textit{Prepare} operation. Thus, any improvements in the implementation of that operation directly translate to an improvement of the \textit{Extend} operation, too. 
\subsubsection{Removing an index and corresponding data element from the QDB}
Similarly to adding new indices to the database, we aim to define an operation that allows removing indices and the respective indexed data element from the database. When the index to remove is not the highest one, we obtain a database with non-continuous indexing. If this is undesirable, one can subsequently apply the \textit{permute} operation to do re-indexing appropriately. In general, the removal operation, which we name~$R_{(f)}$, reduces the number of indices with a nonzero amplitude from the superposition, and the state would be re-normalized. In its simplest form, this operation would, e.g., remove the index state~$f$ and the corresponding data register from the superposition. Thus, by index removal, we understand the following map:
\begin{align}
     \nonumber &\mathcal{H}_I \otimes \mathcal{H}_D  \longrightarrow \mathcal{H}_I \otimes \mathcal{H}_D \\
     &\qdb{k} = \frac{1}{\sqrt{k}}\sum_{j=0}^{k-1} \ket{j}_I \ket{d_j}_D \label{eqn:index_removal} \xmapsto{R_{(f)}} \frac{1}{\sqrt{k-1}} \sum_{\substack{j=0 \\ j \neq f}}^{k-1}  \ket{j}_I \ket{d_j}_D\,,
\end{align}
with $\text{dim}(\mathcal{H}_I) = k$. Similarly to the extension operation presented in~\eqref{eq:extend}, the transformation~$R_{(f)}$ presented in~\eqref{eqn:index_removal} cannot be implemented in general by a unitary transformation. The proof is similar to the one of the extension presented in~\ref{sec:extend_problem} and can be intuitively understood considering the case of removing all indices and their data. This would correspond to an erasure operation that is clearly not unitary since it is a many-to-one function. Thus, the operation depends on the specific database and its indexed orthogonal data elements. We consider the case where the element~$| f \rangle_I |d_f \rangle_D$ is removed. To keep the superposition balanced, we need to take into~account the re-normalization of the amplitudes, i.e.,~$\frac{1}{\sqrt{k}} \mapsto \frac{1}{\sqrt{k-1}}$ as in~\eqref{eqn:index_removal}. In a quantum information setting, the no-cloning theorem applies~\cite{wootters1982single} for a general quantum state. Thereof follows the no-deletion theorem~\cite{pati2006no}. If one aims to remove an index from the superposition, one intends to achieve the following reset operation:
\begin{equation}
    |j\rangle_I |d_j \rangle_D \mapsto |j\rangle_I |0 \rangle_D\,.
\end{equation}
This can be achieved by knowing the state preparation unitary $|d_j \rangle_D = U |0 \dots 0 \rangle_D$. Since in this section we consider classical data in the form of orthogonal data states, another possible option would be to correlate a sensor state equivalent to~$|d_j \rangle_D$ and delete the data in register~$D$ utilizing CNOT gates with additional controls on the index register being in~$j$. This procedure only works for orthogonal states with a known unitary that maps the orthogonal basis states to the computational basis, but this is assumed to be our definition of classical data. Next, the amplitude of the ``resetted'' index state can be transferred back to the~$\ket{0}_{I \otimes D}$ state. This requires knowing the specific index element~$j$ that is removed. This operation has the following effect:
\begin{equation}
    (\ket{0}_I + |j\rangle_I ) \, |0 \rangle_D \mapsto |0\rangle_I |0 \rangle_D \,.
    \label{eqn:reset_index}
\end{equation}
The transformation described in~\eqref{eqn:reset_index} can be achieved by a specific rotation in the two-dimensional space spanned by~$\text{span}\{ \, \ket{0}_I|0  \rangle_D, \, \ket{j}_I|0 \rangle_D \}$. Thus, it is a non-trivial operation to remove elements in the quantum superposition state and requires specific data-dependent operations for each state that fulfills the requirements listed above. Finally, we comment on a non-deterministic implementation that works in the general case but has a finite probability of fully erasing the database. One can measure the index with a binary observable with an eigen-subspace corresponding to the span of the indices to remove and their respective orthogonal subspace. This solution may be acceptable when the database is easy to prepare or multiple copies are available but is not scalable.

\subsubsection{Writing data in the QDB}
\label{sec:writing_data}
In the following, we consider the process of writing individual \textit{data} elements~$d_j$ (with~$1 \leq j \leq (k-1)$) in the data register that are each associated with a specific index in the index register. We assume the data register states associated with the new elements are initially empty, as is the case after the \textit{Extend} operation, e.g., the all-zero state~$|0\rangle_D$. By \textit{writing} the data element~$\ket{d_f}_D$ indexed by~$f$ we understand the following operation:
\begin{align}
    \nonumber 
     &\mathcal{H}_I \otimes \mathcal{H}_D   \longrightarrow \mathcal{H}_I \otimes \mathcal{H}_D \\
     &\qdb{k} = \qdb{k}_{\neq f} + \frac{1}{\sqrt{k}}  \ket{f}_I \ket{0}_D \label{eqn:general_writing} \xmapsto{W_{(f)}} \qdb{k}_{\neq f} + \frac{1}{\sqrt{k}} \ket{f}_I \ket{d_f}_D\,, 
\end{align}
with 
\begin{align}
    |\text{QDB}^{(k)}_{\neq f}\rangle := \frac{1}{\sqrt{k}} \sum_{\substack{j =0 \\ j \neq f}}^{k-1} \ket{j}_I \ket{d_j}_D \, . \nonumber 
\end{align}
For the case of data states drawn from an orthogonal set with a known unitary transformation to the computational basis, we can use a writing procedure based on CNOT gates. Conceptually, we can view the writing process based on the sensor state~$\ket{d_f}_S$ as follows. We transform the orthonormal basis states to the computational basis and apply CNOT gates to copy the data into the superposition of the quantum database. Thus, the operation is schematically given as:
\begin{align}
    W_{(f)} \Big{(}  \Big{(} \qdb{k}_{\neq f} + \frac{1}{\sqrt{k}} \ket{f}_I \ket{0}_D \Big{)} \ket{d_f}_S \Big{)} \label{eqn:schematic_write} \mapsto \left( \qdb{k}_{\neq f} + \frac{1}{\sqrt{k}}  \ket{f}_I \ket{d_f}_D \right) \ket{d_f}_S \, ,
\end{align}
with %
\begin{align}
    W_{(f)} = \left(\left(\mathbb{I}_{I } \otimes U_D \otimes U_D \right) \,\Lambda_{\tau}(\sigma_X)\, \left(\mathbb{I}_{I} \otimes U_D^{\dagger} \otimes U_D^{\dagger} \right)\right) \, . \nonumber
\end{align}
The central operation~$\Lambda_{\tau}(\sigma_X)$ consists of~$\Tilde{m}$ multi-controlled X gates, each controlled on~$\tau :=  \Tilde{k} + 1$ qubits. The gate~$\Lambda_{\tau}(\sigma_X)$ behaves like a CNOT gate with~$\tau$-control qubits. An adapted implementation of~$\Lambda_{\tau}(\sigma_X)$ is presented in~\cite{park2019circuit}. Decompositions of~$\Lambda_{\tau}(\sigma_X)$ are considered in~\cite{Barenco_1995}. The operation defined in~\eqref{eqn:schematic_write} keeps the additional system~$S$ in the product with respect to the database. This allows for tracing out and removing the subsystem~$S$ without any information loss in the database, as no entanglement is created between the QDB and system~$S$. \newline

One may wonder why a conditional SWAP operation is not used to write data in the database~(compare to~\cite{liu2023quantum_memory}, where a similar consideration is done for a \textit{read} operation with respect to a quantum memory, here a quantum processing unit does not necessarily interact with an external system) We note that, upon availability, the dedicated hardware-structure of a QRAM may be used to generate the final state in~\eqref{eqn:schematic_write}. If we condition the SWAP operation on a specific index~$f$, we obtain the following:
\begin{align}
    &W_{(f)} \Big{(} \frac{1}{\sqrt{k}}\Big{(} \sum_{\substack{j=0 \\ j \neq f}}^{k-1} \ket{j}_I \ket{d_j}_D +   \ket{f}_I \ket{0}_D \Big{)}  \ket{d_{f}}_S \Big{)} = \frac{1}{\sqrt{k}} \left( \sum_{\substack{j=0 \\ j \neq f}}^{k-1} \ket{j}_I \ket{d_j}_D  \ket{d_{f}}_S + \ket{f}_I \ket{d_{f}}_D \ket{0}_S \right) \nonumber \\ & \neq \left( \frac{1}{\sqrt{k}}\sum_{\substack{j=0 \\ j \neq f}}^{k-1} \ket{j}_I \ket{d_j}_D + \frac{1}{\sqrt{k}}  \ket{f}_I \ket{d_{f}}_D \right)  \ket{d_{f}}_S \, . 
\end{align}
In general, this state does not remain in product form due to entanglement between the database and system~$S$. Hence, one either accepts a system growth by keeping the ancilla state or takes a process of disentangling~$S$ from the database into account. If only classical correlations are present, one may rely on the disentangling operation introduced in~\cite{cortese2018loading}. Thereby, an ancilla qubit is \textit{decorrelated} from the system by adding an ancilla qubit in the state~$\ket{0}$ and applying a combination of (multi-)controlled CNOT gates. The general process of disentanglement for systems, including non-classical correlations, is further described in~\cite{mor1999disentangling, dodd2004disentanglement, berta2018disentanglement, bandyopadhyay1999disentanglement} and shortly summarized in Appendix~\ref{sec:disentangling_quantum_states}.
\subsubsection{Read-out data in the QDB}
The quantum database contains data elements indexed by orthogonal states that one aims to process in some form. On the one hand, one can directly apply a quantum algorithm on the database state itself, such as Grover's search algorithm for unstructured search~\cite{grover1996fast}, and operate on, and thus, modify the entire QDB state. On the other hand, one may think of a read-out operation that does or does not affect the remaining data states in the QDB. A read-out operation could have the following effect with~$| d_0 \rangle_D 
\neq | d_j \rangle_D $ for some~$0 < j \leq (k-1)$ and an ancilla system with the same dimension as the data register~($\dim(\mathcal{H}_D) = \dim(\mathcal{H}_A)$):
\begin{align}
     \qdb{k} | 0 \rangle_A \not \mapsto   \qdb{k} | d_j \rangle_A\,.
     \label{eqn:read_out} 
\end{align}
One correlates the ancilla system~$A$ to the database in order to transfer the data contained in the $j$th data state to the ancilla system. The transformation in~\eqref{eqn:read_out} is not possible in general due to linearity. Instead, an exemplary copying procedure based on CNOT gates for orthogonal states would result in the following transformation:
\begin{align}
    & \mathcal{H}_I \otimes \mathcal{H}_D \otimes \mathcal{H}_A  \rightarrow \mathcal{H}_I \otimes \mathcal{H}_D \otimes \mathcal{H}_A \nonumber \\ 
     &\qdb{k} | 0 \rangle_A \label{eqn:read_out2} \xmapsto{G_{(f)}}  |\text{QDB}^{(k)}_{\neq f}\rangle | 0 \rangle_A +  \frac{1}{\sqrt{k}}|f \rangle_I | d_f \rangle_D | d_f \rangle_A \, , 
\end{align}
again with 
\begin{align}
    |\text{QDB}^{(k)}_{\neq f}\rangle := \frac{1}{\sqrt{k}} \sum_{\substack{j =0 \\ j \neq f}}^{k-1} \ket{j}_I \ket{d_j}_D \, . \nonumber 
\end{align}
In~\eqref{eqn:read_out2}, the ancilla system~$A$ gets entangled with the quantum database. For orthogonal data states, one could proceed with the copying procedure described in~\eqref{eqn:read_out2} for every data element in the QDB and essentially would create a \textit{copied} data register in a superposition. This state looks as follows:
\begin{align}
    \frac{1}{\sqrt{k}} \sum_{j=0}^{k-1} \ket{j}_I \ket{d_j}_D \ket{d_j}_A \in \mathcal{H}_I \otimes \mathcal{H}_D \otimes \mathcal{H}_A \, .
    \label{eqn:copied_data}
\end{align}
To obtain the state in~\eqref{eqn:copied_data}, one would apply~$\Tilde{m}$ CNOT gates that are controlled by the $i$th qubit in the data register~$D$ and apply the controlled not operation on the $i$th qubit in the ancilla register~$A$. One notices in~\eqref{eqn:copied_data} that the registers $D$ and $A$ are also entangled if at least two states in~$A$ are different, i.e., $\ket{d_j}_A \neq \ket{d_i}_A$ for some $i \neq j$.
\newline
Another possibility to define the read-out of a data state would be via a \textit{projective measurement}. Thereby, we aim to do the projective measurement~\cite{nielsen2001quantum} on the index register and return the corresponding data state indexed by~$f$:
\begin{align}
    |d_f \rangle_D = \frac{\hat{P}_{\ket{f}} \qdb{k}}{\sqrt{\langle \text{QDB}^{(k)}| \hat{P_{\ket{f}}} \qdb{k}}} \, .
    \label{eqn:proj_mmt}
\end{align}
The idempotent projection operator $\hat{P}_{\ket{f}}$ is given by 
\begin{align}
    \hat{P}_{\ket{f}} = \ket{f} \langle f |_I \otimes \mathbb{I}_D \, .
\end{align}
The resulting state of this projective measurement is the data state~$|d_f \rangle_D$ corresponding to index~$f$, which is retrieved with a probability of~$\frac{1}{k}$. Hence, with the given probability of~$\frac{1}{k}$ we obtain the data state~$|d_f \rangle_D$ corresponding to index~$j$. This projective measurement essentially returns a specific data state corresponding to an index and thus forms a query to the quantum database. This query procedure returning the state in~\eqref{eqn:proj_mmt} ``consumes'' the whole database. This is the case as the applied projective measurement collapses the superposition state that forms the quantum database onto a single state. An interesting application could be that the qubits belonging to the index and the qubits belonging to the data registers could be located at different spatial positions. Hence, a projective measurement on the index register would simultaneously collapse the data register at another physical location~\cite{epr_paper}.

\subsubsection{Permute data in the QDB}
\label{sec:permute_elements}
For this operation, we would like to permute database elements according to their indices. We notate this operation by~$P_{\pi}$ and the permutation~$\pi(\cdot) \in S_k$. Thereby, $S_k$ denotes the symmetric group over a finite set of~$k$ elements. We write~$P_{\pi}$ for a quantum database with~$k$ elements as follows:
\begin{align}
    \frac{1}{\sqrt{k}} \sum_{j=0}^{k-1} | j \rangle_I | d_j \rangle_D \xmapsto{P_{\pi}} \frac{1}{\sqrt{k}} \sum_{j=0}^{k-1} | \pi(j) \rangle_I | d_{j} \rangle_D \label{eqn:permutation} = \frac{1}{\sqrt{k}} \sum_{j=0}^{k-1} | j \rangle_I | d_{\pi^{-1}(j)} \rangle_D \, . 
\end{align}
A general permutation~$\pi(\cdot) \in S_k$, or its inverse~$\pi^{-1}(\cdot)$, can always be decomposed into transpositions~$(i, j)$ as two-element transpositions generate the group~$S_k$. We denote a transposition of indices~$(i,j)$ by:
\begin{align}
    \label{eqn:permutation2}
    &i \sigma_x^{i,j} \in \text{SU}(k) 
\end{align}
with the individual entries given by 
\begin{align}
    \left( \sigma_x^{i,j}\right)_{n,m} := \, \delta_{n,m} -  \,\delta_{n,i}  \, \delta_{m,i} - \, \delta_{n,j}  \, \delta_{m,j} +  \,\delta_{n,j} \,  \delta_{m,i} + \, \delta_{n,i}  \, \delta_{m,j} \, .
\end{align}
In the special case with $\text{dim}(I)=4$, one obtains~$P_{(2,3)} = \text{CNOT}$. For every transposition~$(i,j)$ with~$0 \leq i,j \leq (k-1)$, one applies the following operation to the database
\begin{equation}
    P_{(i,j)} := (i\sigma_x^{i,j})_I \otimes \mathbb{I}_D \, .
    \label{eqn:adjacent_transposition}
\end{equation}
By concatenating adjacent transpositions as defined in~\eqref{eqn:adjacent_transposition}, one may generate any permutation~$\pi(\cdot) \in S_k$. A general permutation on the index register~$\pi^{-1}(\cdot)$ as presented in~\eqref{eqn:permutation} can be represented by its permutation matrix $M_{\pi^{-1}} = M_{\pi}^t \in U(k)$. By~$M_{\pi}^t$ we denote the transpose of~$M_{\pi}$. Hence, we have:
\begin{equation}
    P_{\pi} := M_{\pi} \otimes \mathbb{I}_D \in U(k \cdot m) \, .
\end{equation}
As mentioned before, any permutation~$\pi (\cdot)$ can be decomposed into (adjacent) transpositions, which can be implemented by~$P_{(i,i+1)}$ defined in~\eqref{eqn:adjacent_transposition}. It is important to observe that this construction produces a possible decomposition and not the most efficient one. For example, consider the action of a single Pauli-x gate on an index qubit. This changes all indices (by flipping one bit in their binary representation) and may be used as an additional permutation in the decomposition of~$\pi ( \cdot )$ in addition to transpositions. Special cases of the permute operation~$P_{\pi}$ appear in the classical scenario, such as a \textit{shift} or \textit{swap} operation.

\subsection{Outlook towards Quantum Data~(QQ)}
\label{sec:qq}
If we consider the bottom-right quadrant in Figure~\ref{fig:quadrants}, the database contains both quantum indices and quantum data. We defined quantum data as being drawn from a non-orthogonal or orthogonal set of states with an \textit{unknown} unitary transformation to the computational basis. This definition of ``quantum'' data leads to the fact that this type of data generally cannot be copied exactly as perfect cloning is only possible for states belonging to a known set of orthogonal states~\cite{bruss2007lectures}. Thus, for a set of orthogonal states with a known unitary transformation, the respective unitary may act in terms of a \textit{key}. It would allow distinct users to interact differently with the same QDB state depending on their individual access to the specific unitary transformation. 
\newline
Furthermore, from the no-cloning theorem~\cite{wootters1982single} follows the non-deletion theorem~\cite{kumar2000impossibility}. This theorem states that even though information can be deleted perfectly in a reversible manner in classical computation with respect to a copy, the analogous task on quantum information cannot be done for an arbitrary quantum state~\cite{kumar2000impossibility}. It holds due to the linearity of quantum theory~\cite{kumar2000impossibility}. 
\newline
The database preparation does not depend on the type of data written in the respective data register. For the extension operation, the first algorithm based on amplitude amplification only works if we know and are able to implement the unitary that prepares the database state. However, the second version of the database extension algorithm works independently of the data type. 
Regarding the permutation operation, one has to take into account that each data element may be a linear combination of (orthogonal) states. Because of this, data in a superposition corresponds to the same index state. One must consider this fact when conditioning an operation as an individual index transposition on a single data register (corresponding to one element in the \textit{data superposition}).
Also, it is different that data can no longer be removed or written in the QDB, as was the case for orthogonal data states based on CNOT gates. For this, alternative methods that are described below can be considered. 
The Read-out operation differs as the operation in~\eqref{eqn:read_out2} cannot be done for non-orthogonal states; instead, one may consider a swap operation with the following effect:
\begin{align}
    & \mathcal{H}_I \otimes \mathcal{H}_D \otimes \mathcal{H}_A  \rightarrow \mathcal{H}_I \otimes \mathcal{H}_D \otimes \mathcal{H}_A \nonumber \\ 
     &\qdb{k} | 0 \rangle_A \label{eqn:read_out_swap} \xmapsto{G_{(f)}}  |\text{QDB}^{(k)}_{\neq f}\rangle | 0 \rangle_A +  \frac{1}{\sqrt{k}}|f \rangle_I |0 \rangle_D | d_f \rangle_A \, .
\end{align}
As before, the database and the ancilla system~$A$ are entangled in~\eqref{eqn:read_out_swap}. A projective measurement as a read-out operation would return a general quantum state. 

Due to the aforementioned limitations, we highlight possible modifications that allow us to operate on quantum data. Firstly, one may consider the application of approximate or probabilistic cloning that we summarize in the following \cite{bruss2007lectures, qiu2002combinations, lemm2017information}. 
In the case of approximate cloning, one creates an approximate copy by means of a unitary transformation. With regard to the non-trivial approximate cloning strategy, the optimal case would be the universal cloner by \textit{Bu\v{z}ek and Hillery} \cite{buvzek1996quantum}. Thereby, the transformation done by a general unitary is the following~\cite{buvzek1996quantum, scarani2005quantum}:
\begin{align}
    \ket{\psi}_A\ket{R}_B\ket{\mathcal{M}}_M \mapsto \ket{\Psi}_{ABM} \, .
    \label{eqn:approx_cloning}
\end{align}
This forms an optimal and symmetric universal quantum cloner~\cite{buvzek1996quantum, scarani2005quantum} with a single qubit ancilla in this simple example. In~\eqref{eqn:approx_cloning}, the qubit to be cloned is the A-system, the system that contains the approximate clone in the B-system and the ancilla marked by $M$. In the optimal case, this transformation obtains the following~\cite{buvzek1996quantum, scarani2005quantum}:
\begin{align}
    \rho_A = \text{Tr}_{BM} (\ket{\Psi}\bra{\Psi}_{ABM}) = \,\rho_B = \text{Tr}_{AM} (\ket{\Psi}\bra{\Psi}_{ABM}) =\, F\ket{\psi} \bra{\psi} + (1-F) \ket{\psi^{\perp}} \bra{\psi^{\perp}} \, , 
\end{align}
with $F = \frac{5}{6}$. Thus, we have a mixed state if one traces out the ancilla in~\eqref{eqn:approx_cloning}. For the case of the quantum database, this procedure means that the database would get entangled with the qubit containing the data to be written in the database. Thus, either one allows the quantum database system to essentially ``grow,'' or one has to rely on a disentanglement mechanism aiming to disentangle the state from the quantum database~\cite{mor1999disentangling, dodd2004disentanglement, berta2018disentanglement, bandyopadhyay1999disentanglement}.
\newline
Secondly, the case of probabilistic cloning describes the procedure of creating a perfect copy of an unknown quantum state with a success probability smaller than 1. In~\cite{prob_cloning_duan}, it was shown that states that were secretly chosen from a set 
\begin{equation}
    \mathcal{S} = \{ \ket{\psi_1}, \ket{\psi_2}, \dots, \ket{\psi_t}   \}
\end{equation}
can be cloned probabilistically by a general unitary reduction operation if and only if the individual states in $\mathcal{S}$ are linearly independent. Thereby, the unitary transformation~$U$ that is later followed by a post-selective measurement operation looks as follows~\cite{prob_cloning_duan, bruss2007lectures}:
\begin{align}
    U (\ket{\psi_i} \ket{0} \ket{A_0}) = \sqrt{p_i} \ket{\psi_i} \ket{\psi_i} \ket{A_0} \label{eqn:prob_cloning}  +\sum_{j = 1}^n c_{ij} \ket{\Phi_j} \ket{A_j}, \text{ for } i = 1,2, \dots, t \, . 
\end{align}
In~\eqref{eqn:prob_cloning}, the ancilla states are all orthogonal to each other, i.e.,~$\langle A_k | A_l \rangle = \delta_{kl}$ for $0 \leq k,l \leq t$. Therefore, the desired perfect clone is achieved with a probability of~$p_i$ when measuring the ancilla in the basis~$\{A_k\}$. In~\cite{prob_cloning_duan}, a simple example of the cloning efficiency for a linearly independent set of two states~$\mathcal{S} = \{ \ket{\psi_1}, \ket{\psi_2}\}$ is given. In this case, the success probabilities are bound as follows~\cite{prob_cloning_duan, bruss2007lectures}:
\begin{align}
    \frac{1}{2} (p_1 + p_2 ) \leq \frac{1}{1+|\langle \psi_1 | \psi_2 \rangle|}
\end{align}
with~$|\langle \psi_1 | \psi_2 \rangle| \neq 1$.
Let us consider this procedure for writing data items in the data states (with all the data elements being part of a linearly independent set). As a simple example of a one-qubit data register, we have the following database state:
\begin{align}
    \frac{1}{\sqrt{k}} \left(\sum_{\substack{j=0\\ j \neq f}}^{k-1}  \ket{j}_I \ket{d_j}_D + \ket{f}_I \ket{0}_D \right) \ket{d_f} \ket{A_0} \, .
\end{align}
Aiming to write in the data register indexed by~$f$, we apply the unitary presented in~\eqref{eqn:prob_cloning} conditioned on index~$f$ and obtain:
\begin{align}
    \frac{1}{\sqrt{k}} \, \left(\sum_{\substack{j=0\\ j \neq f}}^{k-1}  \ket{j}_I \ket{d_j}_D \ket{d_f} \ket{A_0} + \ket{f}_I (\sqrt{p_i} \ket{d_f}_D  \ket{d_f} \ket{A_0} + \sum_{j = 1}^n c_{ij} \ket{\Phi_j'} \ket{A_j})\right) \, .
\end{align}
If one now post-selectively applies a measurement on the ancilla in the basis~$\{A_k\}_k$, measuring~$\ket{A_0}$ we end up successfully in the subspace containing the QDB with the copy of~$\ket{d_f}$ or failed the cloning when measuring the ancilla in any other state than~$\ket{A_0}$. This means that this ancilla measurement either returns the preferred QDB state or a failure state when measuring~$\{\ket{A_j}\}_j$ for~$\forall j \neq 0$.

\section{Conclusion}
\label{sec:conclusion}
We presented an algorithmic framework for quantum databases that distinguishes between classical and quantum indexing and data items. Our definition of a quantum database consists of a uniform superposition state storing data correlated to specific orthogonal index states in a quantum superposition. We have highlighted essential database operations and presented possible implementations focusing on the case of quantum indexing in order to discuss their feasibility. We note that multiple possible implementations exist for each operation, and we leave it as an open problem to refine and improve each of them individually.

The operations we showcase here already allow for general recommendations for best practices when working on a quantum database. For example, for the index extension operation, differently from the classical indexing case, extending the database is very costly, and one should rather set the total (or maximum) number of items to be added in advance and prepare the database accordingly. Furthermore, the available resources should be exhausted by adding all possible index states during the extension operation. We observe that new index states are released in bunches most efficiently, and one should use that to have a low circuit depth. Furthermore, to unlock quantum advantage from using databases, we recommend avoiding projective measurements in order to keep the superposition state and employ its quantum nature algorithmically. 

In this work, we considered a noise-free scenario. However, amplitude noise would affect the balanced amplitudes in the superposition state, especially when many states are present in the superposition, and the amplitudes get small. Bit-flip errors in the computational basis would disturb the correspondence of data and index elements, but one can rely on efficient error correcting codes, see, e.g.,~\cite{bravyi2024high}, that are expected to improve in the future. A detailed study on the impact of noise in the context of a QDB is left for future work. 

We believe it is important to examine the differences between databases in the classical and quantum scenarios due to the naturally different conditions dictated by quantum mechanics. In this context, one needs to rethink even basic operations, which is essential to formalize a database version based on structured quantum data. The presented framework gives a more precise outlook on the practicality and usefulness of each of the operations. Furthermore, for applying a set of algorithms operating on a quantum database before a state collapse~(e.g., through measurement), the presented operations are applicable as a state preparation and allow for dynamical manipulation after initialization of a certain quantum state. 
Looking into the future, this algorithmic framework of a quantum database could become relevant in use cases in which one aims to efficiently operate on (quantum) data in order to extract global properties of a system, which is of possibly dynamic nature.

\section*{Code availability}
The implementations for the main algorithms of this work have been implemented for verification in \texttt{Python} and \texttt{C++} and are publicly available under the MIT license on \href{https://github.com/carlasophie/QDB_algorithms}{GitHub}.

\section*{Acknowledgements}
The authors would like to thank Lorenzo Maccone for fruitful discussions. CR and MG are supported by CERN through the CERN Quantum Technology Initiative. CR has been sponsored by the Wolfgang Gentner Programme of the German Federal Ministry of Education and Research~(grant no. 13E18CHA). CR would like to thank Tobias Duswald for his support throughout this work. 

\bibliography{mybib}{}
\bibliographystyle{quantum}

\onecolumn
\appendix

\section{Appendix}
\subsection{General implementation of the Prepare operation}
\label{sec:prepare_general}

In this section, we generalize the \textit{prepare} operation described in Section~\ref{sec:prepare}. This generalized algorithm allows for preparing a superposition state with balanced amplitudes that contains a number of elements~$k$ for the case of~$\log_2(k) \notin \mathbb{N}$. Furthermore, we present a modified algorithm that may be used for preparing a balanced or uniform superposition state, with the exception that the zero index~$\ket{0}_I$ is used as a probability reservoir relevant with respect to database extension. Hence, we aim to prepare the general state with~$l \in \mathbb{N}$:
\begin{align}
    \sqrt{\frac{l+1}{k+l}} \, \ket{0}_I+\sqrt{\frac{1}{k+l}} \, \sum_{j=1}^{k-1}\ket{j}_I \, .
    \label{eqn:general_prepared_state}
\end{align}

\begin{figure}[b!]
\begin{algorithm}[H]
\caption{Prepare unbalanced superposition of $k$ computational basis states with extra weight for the $\ket{0}$ state: $\sqrt{l+1}/\sqrt{k+l}\ket{0}+1/\sqrt{k+l}\sum_{j=1}^{k-1}\ket{j}$ .}
\begin{algorithmic}[1]
\Require $k >1, \, k \in \mathbb{N} $ and $l \in \mathbb{N}_{\geq 0} $
\Ensure $ 0 \leq p \leq 1$
\Procedure{Prepare}{$k,l$}\Comment{prepare unbalanced superposition}
\State Compute $t=\lceil \log_2 k \rceil$
\State Express $s=k-1$ in binary form: $s=\sum_{j=0}^{t-1} 2^j s_j$
\Comment{one has $s_{t-1}=1$}
\Statex \Comment{note that bit 0 is the least significant bit}
\State Compute $p=\frac{2^{t-1}+l}{k+l}$
\State Apply $Y(p)$ to qubit $t-1$
\For {$j=t-2, \cdots , 0$}
    \State Apply $Y(1/2)$ to qubit $j$
    \If {$s_j=0$}
        \State Apply $Y(1/2)^{-1}$ controlled on $[s]_{t-1, ..., j+1}$
        \Comment{controlled on qubits $t-1, t-2, ..., j+1$}
        \Statex \Comment{being equal to the corresponding bits of $s$} 
    \Else
        \State Compute $p=\frac{2^j}{[s]_{j, j-1, ..., 0}+1}$
        \State Apply $\tilde{Y}(p)$ controlled on $[s]_{t-1, ..., j+1}$
    \EndIf
    \State Compute $p=\frac{2^j+l}{2^{j+1}+l}$
    \State Apply $\tilde{Y}(p)$ controlled on $[0]_{t-1, ..., j+1}$ \Comment{if $l=0$ then $p=1/2$ and $\tilde{Y}(p)=I$}
\EndFor
\EndProcedure
\end{algorithmic}
\label{alg:prepare}
\end{algorithm}
\end{figure}
\noindent
This state reassembles the completely balanced case if~$l =0$. For creating the state in~\eqref{eqn:general_prepared_state} one can use Algorithm~\ref{alg:prepare}. This procedure includes unitary gates~$Y(p)$ and~$\tilde{Y}(p)$ with~$0 \leq p \leq 1, \, p \in \mathbb{R}$ defined below. The gate~$Y(p)$ plays a similar role to the Hadamard gate for~$p=1/2$ when acting on the~$\ket{0}$ state (i.e.,~$H \ket{0} = Y(1/2) \ket{0}$), but, in general, it corresponds to a~$R_y$ rotation. It is defined as follows:
\begin{equation}
    Y(p) = \begin{pmatrix}
        \sqrt{p} & -\sqrt{1-p} \\
        \sqrt{1-p} & \sqrt{p}
    \end{pmatrix} \, .
\end{equation}
such that~$Y(p) = R_y(\theta)$ for $\theta = 2 \arccos{(\sqrt{p})}$ and~$Y(1) = \mathbb{I}$.
Furthermore, based on the previous definition of~$Y(p)$, we define the following in order to simplify the notation
\begin{equation}
    \tilde{Y}(p) := Y(p)\cdot Y(1/2)^{-1} \, .
\end{equation}
\noindent
The implementation of Algorithm~\ref{alg:prepare} presents the most general case. In order to facilitate a more intuitive understanding, we showcase concrete examples. 
\newline
\newline
\noindent
\textbf{Example case of Algorithm~\ref{alg:prepare} for~$k=22$ and~$l =0$ .} \newline
For this, consider the case of initializing up to~$k=22$ with~$l=0$. Hence, the circuit acts on a 5-qubit register~$I$. The goal in this exemplary case is to create a balanced superposition state
\begin{align}
    \frac{1}{\sqrt{22}} \sum_{j=0}^{21} \ket{j}
\end{align} 
in the register~$I$. The first step is expressing~$k-1$ as a binary number, here~$k-1 = 21 =: b10101$.
Then, we apply unitary gates~$Y(p)$ starting from the qubit of the highest significance, noting that certain gates will be controlled on the higher-significance qubits according to the binary representation of~$(k-1)$. In general, for qubit~$a$, consider the~$a$-th bit of~$k-1$ and denote it by~$(k-1)_a$. The initial layer of one-qubit gates creates a balanced superposition except with regard to the splitting of the most significant bit. In the later stages of the algorithm, balanced splitting is reverted in order to achieve an imbalanced distribution. Thereby, probabilities are distributed according to the final number of branches in the respective binary tree diagram. The corresponding quantum circuit for initializing~$k=22$ is shown in Figure~\ref{fig:circuit_10101}. The circuit depth scales linearly according to the number of qubits with respect to multi-controlled CNOT gates. Several works have been written discussing the decomposition of multi-controlled CNOT gates into CNOT gates and single-qubit gates. State-of-the-art methods discuss its decomposition with~$\tau$ denoting the number of controls. The decompositions scale as~$\mathcal{O}(\tau)$~\cite{gidney2024} or even ~$\mathcal{O}(\text{polylog}(\tau))$ by borrowing an ancilla or considering an approximate case~\cite{claudon2024polylogarithmic}. 
\begin{figure}[h!]
    \centering
    \includegraphics[width=0.45\textwidth]{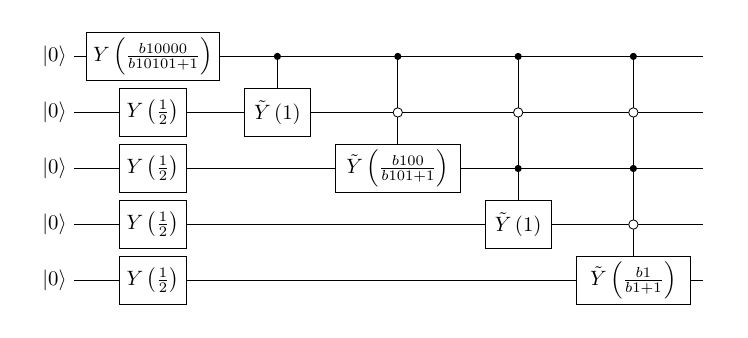}
    \caption{Circuit diagram for creating the balanced quantum superposition state with~$k=22$ elements (where $\log_2(k) \not \in \mathbb{N}$). The presented quantum circuit acts on the index register of the database, and we apply the identity otherwise. The resulting state here is given by~$\tfrac{1}{\sqrt{22}} \sum_{j=0}^{21} \ket{j}$. }
    \label{fig:circuit_10101}
\end{figure}

\noindent
\textbf{Example case of Algorithm~\ref{alg:prepare} for~$k=14$ and some~$l \in \mathbb{N}_{>0}$ .} \newline
Furthermore, we also show an example case for initializing a superposition state with~$l > 0$ and~$k=14$ given by:
\begin{align}
    \sqrt{\frac{l+1}{14+l}} \ket{0} + \frac{1}{\sqrt{14+l}} \sum_{j=1}^{13} \ket{j}.
\end{align} 
We present the binary tree visualization for initializing up to element~$1101$ in Figure~\ref{fig:tree_14}. Each branching step in the binary tree visualization corresponds to applying a~$Y(p)$ gate. If the branching is non-symmetric, symmetric gates must be reverted using controlled~$Y(p)$ gates. The corresponding quantum circuit is shown in Figure~\ref{fig:circuit_1101}. The circuit depth scales linearly according to the number of qubits.

\begin{figure}[h!]
    \centering
    \includegraphics[width=0.55\textwidth]{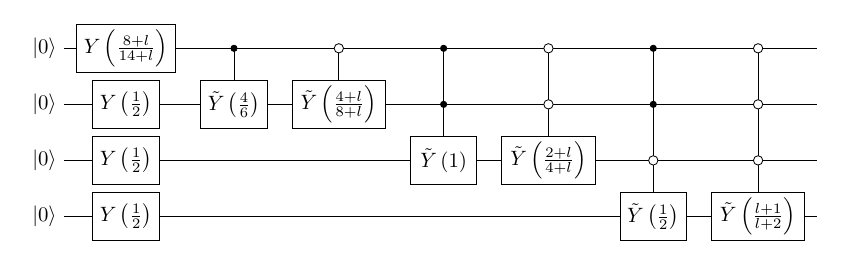}
    \caption{Circuit diagram for creating the  quantum superposition state with~$k=14$ elements (where $\log_2(k) \not \in \mathbb{N}$). The presented quantum circuit acts on the index register of the database, and we apply the identity otherwise. The resulting state here is given by~$ \sqrt{\frac{1+l}{14+l}} \ket{0} +  \sqrt{\frac{1}{14+l}} \sum_{j=1}^{14} \ket{j}$ . }
    \label{fig:circuit_1101}
\end{figure}

\begin{figure}[h!]
    \centering
\scalebox{0.4}{
    \begin{tikzpicture}
    [
        level 1/.style={sibling distance=190mm},
        level 2/.style={sibling distance=95mm},
        level 3/.style={sibling distance=45mm},
        level 4/.style={sibling distance=23mm}
    ]
    	\node {{\color{violet} $1 \cdot$}$\ket{0000}_I$}
    		child {node {\small{\color{violet} $\sqrt{\frac{l+8}{14 + l}} \cdot$}$\ket{0000}_I$}
                child {node {\small{\color{violet} $\sqrt{\frac{l+4}{14 + l}} \cdot$}$\ket{0000}_I$}
                    child {node {\small{\color{violet} $\sqrt{\frac{l+2}{14 + l}} \cdot$}$\ket{0000}_I$}
                        child {node {\small{\color{violet} $\sqrt{\frac{l+1}{14 + l}} \cdot$}$\ket{0000}_I$}}
    		            child {node {\small{\color{violet} $\sqrt{\frac{1}{14 + l}} \cdot$}$\ket{0001}_I$}}
                    }
    		        child {node {\small{\color{violet} $\sqrt{\frac{2}{14 + l}} \cdot$}$\ket{0010}_I$}
                        child {node {\small{\color{violet} $\sqrt{\frac{1}{14 + l}} \cdot$}$\ket{0010}_I$}}
    		            child {node {\small{\color{violet} $\sqrt{\frac{1}{14 + l}} \cdot$}$\ket{0011}_I$}}
                    }
                }
                child {node {\small{\color{violet} $\sqrt{\frac{4}{14 + l}} \cdot$}$\ket{0100}_I$}
                    child {node {\small{\color{violet} $\sqrt{\frac{2}{14 + l}} \cdot$}$\ket{0100}_I$}
                        child {node {\small{\color{violet} $\sqrt{\frac{1}{14 + l}} \cdot$}$\ket{0100}_I$}}
    		            child {node {\small{\color{violet} $\sqrt{\frac{1}{14 + l}} \cdot$}$\ket{0101}_I$}}
                    }
    		        child {node {\small{\color{violet} $\sqrt{\frac{2}{14 + l}} \cdot$}$\ket{0110}_I$}
                        child {node {\small{\color{violet} $\sqrt{\frac{1}{14 + l}} \cdot$}$\ket{0110}_I$}}
    		            child {node {\small{\color{violet} $\sqrt{\frac{1}{14 + l}} \cdot$}$\ket{0111}_I$}}
                    }
                }
            }
    		child {
    		    node {{\color{violet} $\sqrt{\frac{6}{14 + l}} \cdot$}$\ket{1000}_I$}
    		    child {node {{\color{violet} $\sqrt{\frac{4}{14 + l}} \cdot$}$\ket{1000}_I$}
                    child {node {{\color{violet} $\sqrt{\frac{2}{14 + l}} \cdot$}$\ket{1000}_I$}
                        child {node {\small{\color{violet} $\sqrt{\frac{1}{14 + l}} \cdot$}$\ket{1000}_I$}}
    		            child {node {\small{\color{violet} $\sqrt{\frac{1}{14 + l}} \cdot$}$\ket{1001}_I$}}
                    }
    		        child {node {{\color{violet} $\sqrt{\frac{2}{14 + l}} \cdot$}$\ket{1010}_I$}
                        child {node {\small{\color{violet} $\sqrt{\frac{1}{14 + l}} \cdot$}$\ket{1010}_I$}}
    		            child {node {\small{\color{violet} $\sqrt{\frac{1}{14 + l}} \cdot$}$\ket{1011}_I$}}
                    }
                }
    		    child {node {{\color{violet} $\sqrt{\frac{2}{14 + l}} \cdot$}$\ket{1100}_I$}
                    child {node {\small{\color{violet} $\sqrt{\frac{2}{14 + l}} \cdot$}$\ket{1100}_I$}
                        child {node {\small{\color{violet} $\sqrt{\frac{1}{14 + l}} \cdot$}$\ket{1100}_I$}}
    		            child {node {\small{\color{violet} $\sqrt{\frac{1}{14 + l}} \cdot$}$\ket{1101}_I$}}
                    }
    		        child {node {{\color{violet} $0 \cdot$}{\color{gray} $\ket{1110}_I$}}
                        child {node {\small{\color{violet} $0\cdot$}{\color{gray}$\ket{1110}_I$}}}
    		            child {node {\small{\color{violet} $0\cdot$}{\color{gray}$\ket{1111}_I$}}}
                    }
                }
    		};
    \end{tikzpicture}
    }
\caption{Binary tree that visualizes the step-by-step creation of new indices for a database with a probability reservoir. In this case,~$14$ elements are initialized (up to binary element~$1101$) and visualized in a binary tree structure. Each branching point visualizes the application of a (conditional)~$Y(p)$ gate that leads to the creation of a new state per branching, starting from the root towards the tree's leaves.}
    \label{fig:tree_14}
\end{figure}
For the aforementioned exemplary cases we had~$\log_2(k) \not \in \mathbb{N}_{>0}$. If~$\log_2(k) \in \mathbb{N}_{>0}$, the generalized prepare operation acts as the Walsh-Hadamard transformation, see Section~\ref{sec:prepare}. This holds true since~$H \ket{0} = Y(1/2) \ket{0}$. Thus, the algorithm reduces to the balanced prepare operation introduced in Section~\ref{sec:prepare} corresponding to the Walsh-Hadamard transformation, which has a circuit depth of~$\mathcal{O}(1)$. This can be visualized by the circuit identity presented in Figure~\ref{fig:circuit_011}. By applying this circuit, a superposition state with~$4$ elements is prepared.
\begin{figure}[h!]
    \centering
    \includegraphics[width=0.5\textwidth]{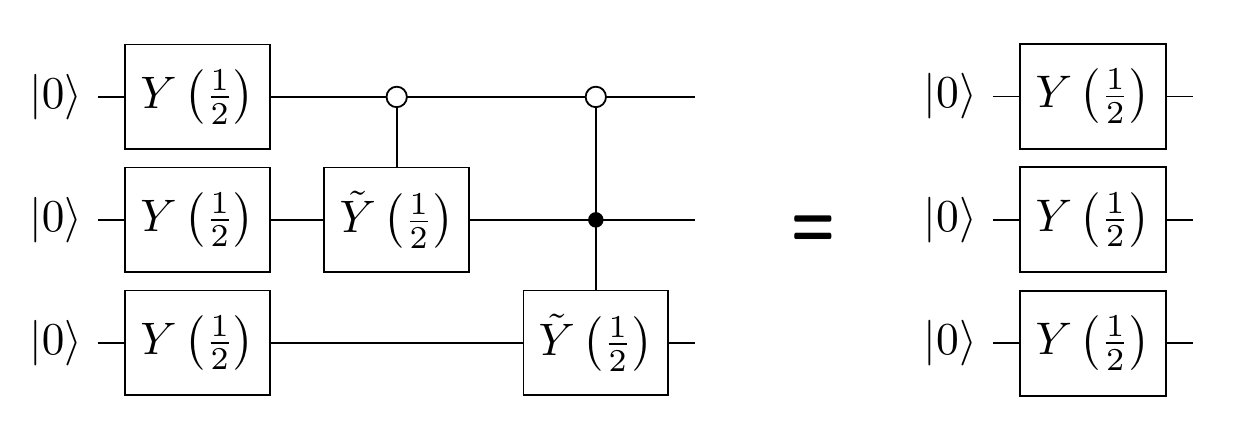}
    \caption{Circuit diagram for creating the  quantum superposition state with~$k=8$ elements (example for $\log_2(k) \in \mathbb{N}$ and~$l =0$). The operation reduces to the Walsh-Hadamard transformation, introduced in the \textit{prepare} operation in Section~\ref{sec:prepare} as the initial state is given by~$\ket{0}^{\otimes 3}$. The presented quantum circuit acts on the index register of the database, and we apply the identity otherwise. The resulting state here is given by~$ \frac{1}{\sqrt{2}} \sum_{j=0}^{7} \ket{j}$ . }
    \label{fig:circuit_011}
\end{figure}

\subsection{Problem with the Extend operation}
\label{sec:extend_problem}

The operation~$E_{(l)}$ described in~\eqref{eq:extend} cannot be implemented by a unitary operation. This holds as the overlap is not preserved by~$E_{(l)}$ (as it would be if~$E_{(l)}$ is unitary) for~$l > 0$. We formally write this observation in the following Lemma. 

\begin{lemma}
    There does not exist a general unitary operation~$E_{(l)}$ with~$l \in \mathbb{N}_{>0}$ that fulfills the transformation in~\eqref{eq:extend} for a general QDB state with~$k\geq 2$ elements.
\end{lemma}

\begin{proof}
    As before, we denote the newly added ancilla by~$A$, and the index and data register by~$I$ and~$D$, respectively. We extend two individual databases containing different data elements, assuming there exists a unitary~$V := E_{(l)}$ as in~\eqref{eq:extend} with the following effect:
    \begin{align}
        \frac{1}{\sqrt{k}} (\ket{0}_I\ket{0}_D + \sum_{j=1}^{k-1} \ket{j}_I\ket{d_j}_D) \xmapsto{V} \frac{1}{\sqrt{k+l}} \left( \ket{0}_A (\ket{0}_I\ket{0}_D + \sum_{j=1}^{k-1} \ket{j}_I\ket{d_j}_D) + \ket{1}_A \sum_{i=0}^{l-1} \ket{i}_I\ket{0}_D \right) \, , \nonumber \\
        \frac{1}{\sqrt{k}} (\ket{0}_I\ket{0}_D + \sum_{j=1}^{k-1} \ket{j}_I\ket{d_j'}_D) \xmapsto{V} \frac{1}{\sqrt{k+l}} \left( \ket{0}_A (\ket{0}_I\ket{0}_D + \sum_{j=1}^{k-1} \ket{j}_I\ket{d_j'}_D) + \ket{1}_A \sum_{i=0}^{l-1} \ket{i}_I\ket{0}_D \right) \, .\label{eqn:unitary_v}
    \end{align}
    If we assume the map in~\eqref{eq:extend}, which has the effect as given in~\eqref{eqn:unitary_v}, can be implemented by a unitary~$V$, the overlap has to be preserved. Thus, we compare the overlap before and after applying $V$ for the two databases for general data states~$\{\ket{d_j}_D, \ket{d_j'}_D \}_j$ and~$\sum_{j=1}^{k-1} \bra{d_j'} d_j \rangle_D < (k-1)$ in~\eqref{eqn:unitary_v}:
    \begin{align}
        \frac{1}{k} \left( 1 + \sum_{j=1}^{k-1} \bra{d_j'} d_j \rangle_D \right) \neq \frac{1}{k+l} \left( 1 + l + \sum_{j=1}^{k-1} \bra{d_j'} d_j \rangle_D \right) \, , \label{eqn:overlap_general}
    \end{align}
    where we took advantage of the orthogonality of the index states.
    As seen in~\eqref{eqn:overlap_general}, the overlap is not preserved except if and only if~$\sum_{j=1}^{k-1} \bra{d_j'} d_j \rangle_D = (k-1)$, meaning that the two databases must be equal. Hence, there does not exist a unitary~$V$ that does the transformation in~\eqref{eq:extend} for an arbitrary database. 
    
\end{proof}

\subsection{Extension of the quantum database size}
\label{sec:extend_database}
In this section, we present in detail the two extension algorithms summarized in Section~\ref{sec:extension_chapter}. We introduce the individual operations and include the specific algorithmic implementation in Algorithm~\ref{alg:combined_transfer} and Algorithm~\ref{alg:imbalanced_extend}. We remind ourselves that there does not exist a general unitary for the extension operation as in~\eqref{eq:extend}; hence, each proposed extension mechanism is data-dependent as seen in Section~\ref{sec:extend_problem}.
\subsubsection{Transfer and unfold}
First, we present a protocol for extending a database state using amplitude transfer. In general, the operation is defined as:
\begin{align}
    \mathcal{H}_I \otimes \mathcal{H}_D &\longrightarrow  \mathcal{H}_{I'} \otimes \mathcal{H}_D := (\mathcal{H}_A \otimes \mathcal{H}_I)_{I'} \otimes \mathcal{H}_D \nonumber \\
    |\text{QDB}^{(k)} \rangle &\mapsto |\text{QDB}^{(k+l)} \rangle \, .
\end{align}
Thereby, the newly created indices $k \leq i < (k+l) $ are all correlated to the empty data string~$\ket{0}_D$. For the first approach, the amplitudes are adjusted so that the database will be a balanced superposition of equally distributed amplitudes after the index creation procedure is applied (i.e., after~$l$ indices are added, with~$0 <l\leq k$). If we tolerate an error in the amplitude distribution, we may apply the general amplitude transfer algorithm presented by Brassard~\textit{et al.}~\cite{brassard2002quantum}. On the other hand, if the amplitudes have to be transferred \textit{exactly}, one applies the \textit{exact} amplitude amplification algorithm~\cite{brassard2002quantum} that includes two additional parameters~$\phi,\rho$. Thus, the former is a particular case of the latter with~$\phi,\rho = \pi$ for the following step operator with the angles~$0 \leq \phi, \rho < 2 \pi$~\cite{brassard2002quantum}:
\begin{align}
\label{eqn:exact_step}
    \mathbf{Q} = \mathbf{Q}(\mathcal{U}_{\qdb{k}}, \chi, \phi, \rho) = - \, \mathcal{U}_{\qdb{k}} \, \mathbf{S}_0 (\phi) \, \mathcal{U}_{\qdb{k}}^{\dagger} \mathbf{S}_{\chi} (\rho)\,.
\end{align}
This is equivalent to the operator for Grover search if the unitary~$\mathcal{U}_{\qdb{k}}$ is substituted by the Walsh-Hadamard transform~\cite{grover1996fast}. In our case, the unitary~$\mathcal{U}_{\qdb{k}}$ corresponds to the operations needed to prepare the database state and satisfies:
\begin{align}
    \mathcal{U}_{\qdb{k}} \left( |0\rangle_{I} |0\rangle_D \right) = \frac{1}{\sqrt{k}} \sum_{j=0}^{k-1}  \ket{j}_I \ket{d_j}_D \, ,
\end{align}
while the operator~$\mathbf{S}_{\chi} (\rho)$ is defined as:
\begin{equation}
    \ket{x} \mapsto \begin{cases}
e^{i \rho}\ket{x}, \, &\text{if} \,\, \chi(x)=1 \, ,\\
\,\,\,\,\,\,\,\, \ket{x}, \, &\text{if} \,\, \chi(x)=0 \, .
\end{cases}
\end{equation}
Thereby,~$\chi(x) = 1$ if and only if~$\ket{x} = \ket{0}_{I \otimes D}$ and~$\chi(x) = 0$ otherwise. Furthermore,~$\mathbf{S}_0 (\phi)$ is acting as a phase multiplication by~$e^{i \phi}$ to the all-zero string and equals to the identity operation otherwise. Hence, in the approximate case, one would apply the step operator $m$-times (i.e., this is given by~$(\mathbf{Q}(\mathcal{U}_{\qdb{k}}, \chi, \pi, \pi))^m$) with~$m = \lfloor m^* \rfloor$ and $\lfloor x \rfloor$ denoting the floor function that returns the greatest integer~$\leq x$ for~$x \in \mathbb{R}$. Thereby,~$m^* \in \mathbb{R}$ is given by the relation:
\begin{equation}
    m^* =\frac{ \mp \arcsin\left(\frac{1}{\sqrt{k+l}}\right)- \arcsin \left(\frac{1}{\sqrt{k}}\right)+ \pi  n}{2 \arcsin \left(\frac{1}{\sqrt{k}}\right)}
\end{equation}
with~$n\in \mathbb{Z}$,~$0 < l \leq k$ and~$k \geq 2$ such that~$m^*>0$. As the step operator~$(\mathbf{Q}(\mathcal{U}_{\qdb{k}}, \chi, \pi, \pi))^m$ is only applied in integer multiples, the last step is done with respect to the generalized step operator given in~\eqref{eqn:exact_step} to obtain an exact amplitude transfer. To transfer the correct amount of amplitude, one can make arbitrarily small rotations by varying the angles~$\rho, \psi$ in~\eqref{eqn:exact_step}. This is because if parameters~$\rho, \psi$ are varied continuously between~$0$ and~$2\pi$. Hence, to achieve a correct amplitude transfer in the last step of the operation, the parameters are chosen accordingly to fulfill the following equation:
\begin{align}
     \langle 0 | (\mathbf{Q}(\mathcal{U}_{\qdb{k}}, \chi, \phi, \rho)\cdot(\mathbf{Q}(\mathcal{U}_{\qdb{k}}, \chi, \pi, \pi))^m) |0 \rangle_{I \otimes D} = \sqrt{\frac{l+1}{k +l}} \,. \label{eqn:cond_chi}
\end{align}
The normalization condition implies:
\begin{align}
     \langle \text{QDB'}^{(k)}_{\neq \, 0} | (\mathbf{Q}(\mathcal{U}_{\qdb{k}}, \chi, \phi, \rho) \cdot (\mathbf{Q}(\mathcal{U}_{\qdb{k}}, \chi, \pi, \pi))^m) |\text{QDB'}^{(k)}_{\neq \, 0} \rangle = \sqrt{\frac{k-1}{k +l}} \, , 
\end{align}
with 
\begin{align}
    |\text{QDB'}^{(k)}_{\neq \, 0} \, \rangle := \frac{1}{\sqrt{k-1}} \sum_{j=1}^{k-1} \ket{j}_I \ket{d_j}_D 
\end{align}
being the re-normalized superposition state of~$k$ indices and data elements without the state~$\ket{0}_{I \otimes D}$. The best practice for the index extension scenario is to add as many indices as possible and thus be resource-efficient concerning each newly added ancilla qubit. In this case, we have~$l = k$. Thus, the total number of indices doubles during the operation. After the amplitude transfer is applied, the imbalanced amplitudes are now given in the exact case:
\begin{align}
     |\psi_1 \rangle_{I \otimes D} := \sqrt{\frac{l+1}{k+l}} \ket{0}_{I \otimes D} + \frac{1}{\sqrt{k+l}} \sum_{j=1}^{k-1}  \ket{j}_I \ket{d_j}_D  
     = \sqrt{\frac{l+1}{k+l}} \ket{0}_{I \otimes D} + \sqrt{\frac{k-1}{k+l}} |\text{QDB'}^{(k)}_{\neq \, 0} \rangle \,.
\end{align}

\begin{figure}[h!]
\begin{algorithm}[H]
\caption{\label{alg:transfer_ampl1} Transfer amplitude to zero-index state of $\qdb{k}$ (an application of Quantum Amplitude Amplification \cite{brassard1997exact, brassard2002quantum}).}
\begin{algorithmic}[1]
\Require $k >1, \, k \in \mathbb{N} $ and $l, m \in \mathbb{N}_{>0} $
\Procedure{Transfer}{$k,l$}\Comment{transfer amplitude to the $\ket{0}_I$}
\State $\qdb{k} \leftarrow \frac{1}{\sqrt{k}} \sum_{j=0}^{k-1}  \ket{j}_I \ket{d_j}_D  $ \Comment{insert $\qdb{k}$}
\State  $ |\text{QDB}^{(k)}_{imb}\rangle \leftarrow \left(\mathbf{Q}\, (\mathcal{U}_{\qdb{k}}, \chi, \phi=\pi, \rho=\pi) \right)^m \, \qdb{k} $ \Comment{apply $m$ times}
\State $|\text{QDB}^{(k)}_{imb}\rangle \leftarrow \left( \mathbf{Q} \, (\mathcal{U}_{\qdb{k}}, \chi, \phi, \rho) \right) |\text{QDB}^{(k)}_{imb}\rangle$ \Comment{apply with $\chi, \rho$ as in~\eqref{eqn:cond_chi}}
\EndProcedure
\end{algorithmic}
\end{algorithm}
\end{figure}
As the amplitude transfer is achieved, an ancilla qubit in the state~$|0\rangle$ is added as a next step. A conditional $Y$-rotation gate is applied that acts non-trivially on the added ancilla qubit. The $Y$-rotation is conditioned on the index being the zero-string and the angle~$\theta$ is given by~$\theta = 2 \arccos{(\frac{1}{\sqrt{l+1}})}$. The operation is explicitly given by: 
\begin{align}
     CR_y(\theta) := R_y(\theta) \otimes |0 \rangle \langle 0 |_{I \otimes D} + \mathbb{I}_A \otimes \left( \mathbb{I}_{I \otimes D} - |0 \rangle \langle 0 |_{I \otimes D} \right). \label{eqn:cry}
\end{align}
Applying this operation from~\eqref{eqn:cry}, one obtains the state:
\begin{align}
\label{eqn:cry_acting}
     CR_y(\theta) &\left( \ket{0}_A |\psi_1 \rangle_{I \otimes D} \right) = 
       \Bigg{(}\sqrt{\frac{1}{k+l}}|0\rangle_A + \sqrt{\frac{l}{k+l}}|1\rangle_A \Bigg{)} \ket{0}_{I \otimes D} + \frac{1}{\sqrt{k+l}} \sum_{j=1}^{k-1}  \left(|0\rangle_A \ket{j}\right)_I \ket{d_j}_D \, . 
\end{align}
The state given in~\eqref{eqn:cry_acting} can be rewritten as:
\begin{align}
     \sqrt{\frac{l}{k+l}} |1 \rangle_A \ket{0}_{I \otimes D} + \frac{1}{\sqrt{k+l}} \sum_{j=0}^{k-1}  \left(|0\rangle_A \ket{j}_I\right)_{I'} \ket{d_j}_D \, .
     \label{eqn:state1}
\end{align}
Thereby, we used the following notation~$\ket{0}_{I \otimes D} := \ket{ 0}_I \ket{0}_D$ as before. Given the state in~\eqref{eqn:state1}, the index of the zero-index~$I$ part in the state~$ |1\rangle_A |0 \rangle_I|0 \rangle_D$ is \textit{unfolded}. If we read the index basis from left to right, the new indices correlated to empty data strings will be all \textit{odd} numbers. To create~$l$ new \textit{odd} indices correlated to the empty data state~$\ket{0}_D$, we apply the following operation:
\begin{align}
    \label{eq:create_ind_alg1_general}
    |1 \rangle \langle 1|_A \otimes (P_{(l)})_I \otimes \, \mathbb{I}_D + |0 \rangle \langle 0|_A \otimes \mathbb{I}_I \otimes \, \mathbb{I}_D \, .
\end{align}
As before, the operation~$P_{(l)}$ reduces to the Walsh-Hadamard transformation for~$\log_2 (l) \not \in \mathbb{N}_{>0}$. One observes that it makes sense to fully use the available resources by adding all possible new indices per added ancilla qubit with regard to the circuit depth for each index extension operation. The circuit depth of~$P_{(l)}$ is~$\mathcal{O}(1)$ if $\log_2 (l) \in \mathbb{N}_{>0}$ and scales linearly for $\log_2 (l) \not \in \mathbb{N}_{>0}$ according to~$\mathcal{O}(n)$ with~$n = \lceil \log_2 (l) \rceil$.

The extension operation presented here is expensive as the step operator in~\eqref{eqn:exact_step} includes the database creation unitary~$\mathcal{U}_{\qdb{k}}$, and its inverse~$\mathcal{U}_{\qdb{k}}^{\dagger}$, that one might not even have at hand. One has to apply this unitary repeatedly if~$m > 1$. Hence, we considered an alternative approach that requires setting the maximal number of indices to be added beforehand. In the next section, we discuss a protocol in which one initializes a specifically imbalanced database during the preparation stage of the empty database and creates new indices later. 

\begin{figure}[h!]
\begin{algorithm}[H]
\caption{Unfold zero index-string based on Walsh-Hadamard (WH) transformation with one ancilla qubit.}\label{alg:unfold_index}
\begin{algorithmic}[1]
\Require $k >1, \, k \in \mathbb{N} $ and $l \in \mathbb{N}_{>0} $
\Procedure{Unfold}{$k,l$}\Comment{unfold zero index string into $l$ new index states}
\State $|\text{QDB}^{(k)}_{imb}\rangle$ \Comment{insert imbalanced QDB}
\State $|\text{QDB}^{(k)}_{imb}\rangle \leftarrow |0\rangle_A |\text{QDB}^{(k)}_{imb}\rangle$ \Comment{add $|0\rangle_A$}
\State  $|\text{QDB}^{(k)}_{imb}\rangle \leftarrow CR_y (\theta)|\text{QDB}^{(k)}_{imb}\rangle$ \Comment{$CR_y (\theta)$ as in \eqref{eqn:cry}}
\State $ |\text{QDB}^{(k+l)}_{imb}\rangle \leftarrow \left( |1 \rangle \langle 1|_A \otimes (P_{(l)})_I \otimes \, \mathbb{I}_D + |0 \rangle \langle 0|_A \otimes \mathbb{I}_{I \otimes D} \right)|\text{QDB}^{(k)}_{imb}\rangle$ \Comment{apply $P_{(l)}$ on register~$I$}
\EndProcedure
\end{algorithmic}
\end{algorithm}
\end{figure}
\begin{figure}[h!]
\begin{algorithm}[H]
\caption{Extension of $\qdb{k}$ by arbitrary $l>0$.}\label{alg:combined_transfer}
\begin{algorithmic}[1]
\Require $k >1, \, k \in \mathbb{N} $ and $l \in \mathbb{N}_{>0} $
\Procedure{Extend}{$k,l$}\Comment{extend database by $l$ indices correlated to~$\ket{0}_D$}
\While{$l>k$} 
\State Set $l' \leftarrow k$ and $l \leftarrow (l-k)$
\State Apply TRANSFER($k,l'$) \Comment{amplitude transfer using Alg.~\ref{alg:transfer_ampl1}}
\State Apply UNFOLD($k,l'$) \Comment{create new index states by using Alg.~\ref{alg:unfold_index}}
\State $k \leftarrow k+l'$
\EndWhile
\If{$l \leq k$} 
\State Set $l' \leftarrow l$
\State Apply TRANSFER($k,l'$) \Comment{amplitude transfer using Alg.~\ref{alg:transfer_ampl1}}
\State Apply UNFOLD($k,l'$) \Comment{create new index states by using Alg.~\ref{alg:unfold_index}}
\State $k \leftarrow k+l'$
\EndIf 
\State \textbf{return} $\qdb{k}$\Comment{balanced superposition as database with $k+l$ index elements}
\EndProcedure
\end{algorithmic}
\end{algorithm}
\end{figure}

\subsubsection{Initialize and unfold}
This extension technique does not depend on using the unitary~$\mathcal{U}_{\qdb{k}}$ (or its inverse). In general, the extension mechanism is described by the following map:
\begin{align}
    \mathcal{H}_I \otimes \mathcal{H}_D &\longrightarrow  \mathcal{H}_{I'} \otimes \mathcal{H}_D := (\mathcal{H}_A + \mathcal{H}_I)_{I'} \otimes \mathcal{H}_D \nonumber \\
    |\text{QDB}^{(k)}_{imb} \rangle &\mapsto |\text{QDB}^{(k+f)}_{imb'} \rangle \, .
\end{align}
The dimensions are given as follows:~$\dim(\mathcal{H}_A) = z$,~$\dim(\mathcal{H}_I) = k$ and~$\dim(\mathcal{H}_D) = m$ as before. We make use of a modification of the generalized prepare operation~$P_{(l)}$ given in~\ref{sec:prepare_general} or alternatively general amplitude encoding~(see e.g.,~\cite{schuld2021machine}). Thereby, we know that we would like to add at most~$l$ new elements to the database. Thus, we initialize the following state:
\begin{align}
    |\text{QDB}^{(k)}_{imb} \rangle := \sqrt{\frac{l+1}{l+k}} \, \ket{0}_{I \otimes D} + \sqrt{\frac{1}{l+k}} \, \sum_{j =1}^{k-1} |j \rangle_I | 0 \rangle_D \, .
\end{align}
A corresponding quantum circuit for amplitude encoding operates on~$\Tilde{k}$ qubits and has a depth of~$\mathcal{O}(\text{poly}(k))$~\cite{schuld2021machine}. Following this operation, we write the first~$(k-1)$ data elements to the empty database we just prepared, and we obtain the following:
\begin{align}
    |\text{QDB}^{(k)}_{imb'} \rangle =
     \sqrt{\frac{l+1}{l+k}} \, | 0 \rangle_{I \otimes D} + \sqrt{\frac{1}{l+k}} \, \sum_{j =1}^{k-1} |j \rangle_I | d_j \rangle_D \, .
\end{align}
Next, we correlate~$z$ ancilla qubits in the~$\ket{0}$-state, denoted by~$|0\rangle_A$, in order to create at most~$l \leq (2^z-1) \cdot k$ new indices correlated to the empty data string: 
\begin{align}
    |0\rangle_A \left(\sqrt{\frac{l+1}{l+k}} \, | 0 \rangle_{I \otimes D} + \sqrt{\frac{1}{l+k}} \, \sum_{j =1}^{k-1} |j \rangle_I | d_j \rangle_D \right)\,.
    \label{eqn:zancilla0}
\end{align}
In this case, we can read the binary indices from right to left, and thus, all newly added indices have an absolute value larger than~$k$. Hence, if we add new indices in a parallelized manner, we proceed as follows for~$l' \leq (2^z -1)$:
\begin{align}
     CP_A (z) := \Big{(} (P_{(l'+1)})_A \otimes |0 \rangle \langle 0 |_I  +  \mathbb{I}_{A} \otimes (\mathbb{I}_I - |0 \rangle \langle 0|_I )\Big{)} \otimes \mathbb{I}_D \, .
    \label{eqn:hadamard_op}
\end{align}
The best practice here is to create as many new indices as possible. This reduces the operation~\eqref{eqn:hadamard_op} to:
\begin{align}
     CH_A (z) := \Big{(} (H^{\otimes z})_A \otimes |0 \rangle \langle 0 |_I  +  \mathbb{I}_{A} \otimes (\mathbb{I}_I - |0 \rangle \langle 0|_I )\Big{)} \otimes \mathbb{I}_D \, ,
    \label{eqn:hadamard_op22}
\end{align}
and we have~$l' = (2^z -1)$ when using~\eqref{eqn:hadamard_op22}. Thereby,~$l'$ is chosen to be ~$l'=(2^z -1)$ if~$ l \geq (2^z -1)$ and~$l'=l$ otherwise. Applying the operator given in~\eqref{eqn:hadamard_op} on the database state in~\eqref{eqn:zancilla0} results in:
\begin{align}
     \sqrt{\frac{l+1}{(l'+1)\cdot (l+k) }} \, \sum_{i=0}^{l'}|i\rangle_A | 0 \rangle_{I \otimes D} + \sqrt{\frac{1}{l+k}} \, |0\rangle_A \sum_{j =1}^{k-1} |j \rangle_I | d_j \rangle_D \,.
     \label{eqn:qdb_unfold1}
\end{align}
 By applying~\eqref{eqn:hadamard_op} on the state~\eqref{eqn:zancilla0} we bring the indices in the~$A$-subspace in superposition, conditioned on the index in~$I$ being~$|0 \rangle_I$.
 As given in~\eqref{eqn:qdb_unfold1}, the operation decreases the amplitude of the initial $\ket{0}_{I \otimes D}$ state by a factor of~$\sqrt{1/(l'+1)}$. We created a superposition of all possible combinations of orthogonal states in the $A$-subspace that are all in a product state concerning the empty old index string part~$I$ and the all-zero data part, i.e., the state~$\ket{0}_{I \otimes D}$. The operation defined in~\eqref{eqn:hadamard_op} thus creates~$l'$ new index states correlated to the empty data register. Alternatively, if one aims to add more indices, one may bring the old index $I$-subspace (correlated to indices with a Hamming weight larger than one in the $A$-subspace) into superposition. In the case that~$l\leq (2^z -1)$, we do not need to apply the next steps as we have already added enough indices. If~$l> (2^z -1)$, we add~$l-l' =: l'' \leq (k-1) \cdot (2^z-1)$ new indices. For this, we mark all states with a Hamming weight larger than one by an ancilla qubit. This is done by utilizing a multi-controlled CNOT gate~$\Lambda_{\tau}(\sigma_x)$ with $\tau$ controls by applying the following with the initial state given by~$\ket{0}_A \ket{\mathbf{s}}_B \in \mathcal{H}_A \otimes \mathcal{H}_{B}$ with~$\mathbf{s} \in \{0,1 \}^{|B|}$:
 \begin{equation}
      \left( \mathbb{I}_A \otimes\left(\otimes_{|B|} \sigma_x \right) \right)  \left( \Lambda_{|B|}(\sigma_x) \ket{0}_A \ket{\Bar{\mathbf{s}}}_B \right) \left( \mathbb{I}_A \otimes\left(\otimes_{|B|} \sigma_x \right) \right) \, .
 \end{equation}
 thereby,~$\mathcal{H}_A$ denotes the ancilla subspace, and the control is applied onto~$\mathcal{H}_B$. The resulting state is given by~$\ket{f(\mathbf{s})} \ket{\mathbf{s}}$ with
 \begin{equation}
    f(\mathbf{s}) = \begin{cases} 0, \,\,\text{if H}(\mathbf{s}) = 0 \,, \\ 1, \,\, \text{if H}(\mathbf{s}) > 0 \,. \end{cases}
    \label{eqn:hamming_func} \nonumber
 \end{equation}
 In~\eqref{eqn:hamming_func},~$\text{H} (\mathbf{s})$ denoting the Hamming weight of the binary string~$\mathbf{s}$. Therefore, a $\sigma_x$ gate is applied whenever the control is given by~$\mathbf{s} = \{0\}^{|B|}$. The scaling of the algorithm depends on the decomposition of the gate~$\Lambda_{\tau}(\sigma_x)$. Its decomposition into CNOT and single-qubit gates has been discussed in Appendix~\ref{sec:prepare_general}. Conditioned on the ancilla qubit being in~$\ket{0}$, we create new indices in the \textit{old} index subspace~$I$. Alternatively to adding a marker qubit, one can apply the following conditional operation with the effect (with~$l' \leq (\log_2(k) -1)$):
 \begin{align}
    CP_I (k) := \Big{(}|0 \rangle \langle 0 |_A \otimes \mathbb{I}_{I \otimes D} \Big{)} + \Big{(} ( \mathbb{I}_A -  |0 \rangle \langle 0 |_A) \otimes (P_{(l'')})_I \otimes  \mathbb{I}_D \Big{)} \, .
    \label{eqn:hadamard_op2}
\end{align}
This transformation decreases the amplitude of all states (except those correlated to~$|0 \rangle_A$) by a factor of~$\frac{1}{\sqrt{l''}}$. As a result, we added in total~$l' \cdot l''$ new index states correlated to the empty data string. A complete amplitude balance is only achieved in a special case for this operation. The final state after applying this extension procedure and adding all the indices is given as follows:
\begin{align}
    \label{eqn:qdb_extend2}
    |\text{QDB}^{(k + l' \cdot l'')}_{imb} \rangle := \alpha \cdot |0\rangle_A | 0 \rangle_{I \otimes D} + \sum_{j =1}^{k-1} \beta \cdot (|0\rangle_A |j \rangle_I) | d_j \rangle_D  + \sum_{h = 0}^{l''-1} \sum_{j =1}^{l'+1} \gamma \cdot  |j \rangle_{A}|h \rangle_{I} |0 \rangle_D \, ,
\end{align}
with 
\begin{align}
    \alpha := \sqrt{\frac{l+1}{(l'+1)\cdot (l+k) }}, \,\,\, \beta := \sqrt{\frac{1}{l+k }}, \,\,\, \gamma := \sqrt{\frac{l+1}{l'' \cdot (l'+1) \cdot (l+k)}} \, .
\end{align}
Therefore, we obtain a database that is balanced except with respect to the~$\ket{0}_A| 0 \rangle_{I \otimes D} $ state, if~$\frac{l+1}{(l'+1)\cdot l''} \equiv 1$.
\begin{figure}
\begin{algorithm}[H]
    \caption{Extend imbalanced database by $l$ new index elements with multiple ancilla qubits.}\label{alg:imbalanced_extend}
    \begin{algorithmic}[1]
        \Require $k >1, \, k \in \mathbb{N} $ and $l, l', l'', z \in \mathbb{N}_{>0} $
        \Procedure{Imbalanced Extend}{$k,l$}\Comment{extend imbalanced $|\text{QDB}^{(k)}_{imb}\rangle$}
        \State $|\text{QDB}^{(k)}_{imb}\rangle \leftarrow \sqrt{\frac{l+1}{k+l}} \, |0\rangle_{I\otimes D} + \frac{1}{\sqrt{k+l}} \sum_{j=1}^{k-1} |j\rangle_I \otimes \, |0 \rangle_D$
        \Comment{prepare empty QDB}
        \State $|\text{QDB}^{(k)}_{imb}\rangle \leftarrow \sqrt{\frac{l+1}{k+l}} \, |0\rangle_{I\otimes D} + \frac{1}{\sqrt{k+l}} \sum_{j=1}^{k-1} |j\rangle_I \otimes \, | d_j\rangle_D$  \Comment{write data in register $D$}
        \State $|\text{QDB}^{(k)}_{imb}\rangle \leftarrow |0 \rangle_A |\text{QDB}^{(k)}_{imb}\rangle$ \Comment{correlate $z$ ancilla qubits to index}
        \State $|\text{QDB}^{(k+l')}_{imb}\rangle \leftarrow CP_A (z) |\text{QDB}^{(k)}_{imb}\rangle$ \Comment{create $(l')$ new indices}
        \State $|\text{QDB}^{(k + l' \cdot l'')}_{imb}\rangle \leftarrow CP_I (k) |\text{QDB}^{(k+l')}_{imb}\rangle$ \Comment{create $(l')\cdot (l''-1)$ new indices}
        \EndProcedure
    \end{algorithmic}
\end{algorithm}
\end{figure}
As a note, we add here that, in general, one may add individual indices instead of applying~\eqref{eqn:qdb_extend2} based on a combination of Hadamard and CNOT operations. When using this, one must keep track of the newly added indices to avoid adding the same index \textit{twice}. This would ``disturb'' the correlated data state. 
\newline
\newline
\noindent
\textbf{Example case of Algorithm~\ref{alg:imbalanced_extend} for~$z=1$ and~$l \leq k$ .} \newline
\label{sec:extend2_simple}
\noindent
We present the extension Algorithm~\ref{alg:imbalanced_extend} for the simplified case in which we add one ancilla qubit~($z=1$) to the index register of the imbalanced database. Initially, we start off in an imbalanced database state (prepared for~$l$ new elements to be added) and have a large amplitude on the all-zero index state that we see as playing the role of a \textit{reservoir}. The database contains already written data elements; thus, we start off in the imbalanced state that is prepared for~$l$ new index elements to be added:
\begin{equation}
    |\text{QDB}^{(k)}_{imb}\rangle = \sqrt{\frac{l+1}{l+k}}\ket{0}_I \ket{0}_D + \frac{1}{\sqrt{l+k}}\sum_{j=1}^{k-1} \ket{j}_I \ket{d_j}_D \, . 
\end{equation}
Next, we add a single ancilla to the index register. We denote the ancilla register by~$A$:
\begin{equation}
    |\text{QDB}^{(k)}_{imb,A}\rangle = \sqrt{\frac{l+1}{l+k}}\ket{0}_A \ket{0}_I \ket{0}_D + \frac{1}{\sqrt{l+k}}\sum_{j=1}^{k-1}\ket{0}_A \ket{j}_I \ket{d_j}_D \, . 
    \label{eq:imb_A}
\end{equation}
Following that, we apply a conditional $Y$-rotation to the state in~\eqref{eq:imb_A} and obtain:
\begin{align}
    |\text{QDB}^{(k)}_{imb,A}\rangle = \sqrt{\frac{l+1}{l+k}}\left( \frac{1}{\sqrt{l+1}}\ket{0}_A + \sqrt{\frac{l}{l+1}} \ket{1}_A \right) \ket{0}_I \ket{0}_D + \frac{1}{\sqrt{l+k}}\sum_{j=1}^{k-1}\ket{0}_A \ket{j}_I \ket{d_j}_D \,  \nonumber \\
    = \frac{1}{\sqrt{l+k}}\ket{0}_A\ket{0}_I \ket{0}_D + \sqrt{\frac{l}{l+k}} \ket{1}_A \ket{0}_I \ket{0}_D + \frac{1}{\sqrt{l+k}}\sum_{j=1}^{k-1}\ket{0}_A \ket{j}_I \ket{d_j}_D \, . 
    \label{eq:imb_A2}
\end{align}
Following that, we can create new indices from the state~$\ket{1}_A \ket{0}_I$ in order not to affect the already existing indices that correspond to written data elements. To implement this function, we use the generalized prepare operation for~$l$ elements:
\begin{align}
    \ket{1}_A \ket{0}_I \xmapsto{P_{(l)}} \ket{1}_A \left(\frac{1}{\sqrt{l}} \sum_{j=0}^{l-1}\ket{j}_I \right) \, .
\end{align}
Thus, we obtained the following superposition state with~$l$ new indices:
\begin{align}
    \qdb{k+l} = \frac{1}{\sqrt{l+k}}\ket{0}_A\ket{0}_I \ket{0}_D + \frac{1}{\sqrt{l+k}} \sum_{i=0}^{l-1} \left( \ket{1}_A \ket{i}_I \right) \ket{0}_D + \frac{1}{\sqrt{l+k}}\sum_{j=1}^{k-1}\ket{0}_A \ket{j}_I \ket{d_j}_D \, .
    \label{eq:balanced_kl}
\end{align}
In~\eqref{eq:balanced_kl}, we obtained a balanced quantum database state with~$l$ new index elements correlated to~$\ket{0}_D$.

\subsection{Note on disentangling quantum states}
\label{sec:disentangling_quantum_states}
In order to describe the process of~\textit{disentanglement}, we first summarize \textit{entanglement} conceptually. Hence, we consider the Schmidt decomposition for bipartite pure states of a composite system~$A+B$. The Schmidt decomposition~\cite{nielsen2001quantum} of a pure quantum state is given by:
\begin{equation}
    | \psi \rangle = \sum_{i=1}^{r_{\psi}} \lambda_i (|a_i \rangle_A \otimes |b_i \rangle_B )
\end{equation}
with~$\lambda_i \in \mathbb{R}, \lambda_i \geq 0$ and~$\sum_{i=1}^{r_{\psi}} \lambda_i^2 =1$. Thereby,~$r_{\psi}$ is the Schmidt rank of~$| \psi \rangle$, and a state is entangled if and only if~$r_{\psi}>1$. The reduced density matrices of the subsystems $A$ and $B$ are given by~$\text{Tr}_B(| \psi \rangle \langle \psi |) =: \rho_A = \sum_{i=1}^{r_{\psi}} \lambda_i^2 (|a_i \rangle \langle a_i|_A)$ and~$\text{Tr}_A(| \psi \rangle \langle \psi |) =: \rho_B = \sum_{i=1}^{r_{\psi}} \lambda_i^2 (|b_i \rangle \langle b_i|_B)$.
During the general read-in process to the quantum database, a respective sensor would get entangled with the database state. However, after the data is stored in the database, we no longer need the sensor and would like to trace out this additional system. This is non-trivial due to the entanglement between both subsystems (database and sensor). Thus, we consider the concept of disentanglement while preserving as many local properties as possible. In general, in the case of the database, we aim to maintain the database but not the sensor system. \newline
\noindent 
Under the process of \textit{disentanglement}, as introduced in~\cite{mor1999disentangling, bruss1998optimal}, we understand the following operation:
\begin{theorem}[Definition~\cite{mor1999disentangling, bandyopadhyay1999disentanglement}]
Disentanglement is the process that transforms a state of two or more subsystems into an unentangled state in general, i.e., a mixture of product states such that the reduced density matrices of each of the subsystems are unaffected.
\end{theorem}
\noindent 
Mathematically, we aim to do the following transformation~\cite{terno1999nonlinear}:
\begin{equation}
    \rho_{AB} = | \psi \rangle \langle \psi |  \rightarrow \rho_A \otimes \rho_B \, .
    \label{eqn:distangling_machine}
\end{equation}
It was shown that there exists no universal disentangling machine (describing the transformation for a general entangled state~$\rho_{AB}$ in~\eqref{eqn:distangling_machine}) that transforms an arbitrary entangled (or inseparable state) into a separable one~\cite{mor1999disentangling, terno1999nonlinear}. 
But what happens if we are only interested in a unitary transformation doing the following operation :
\begin{equation}
    | \psi \rangle \langle \psi |_{AB} = \sum_{i,j=1}^{r_{\psi}} \lambda_i \lambda_j^* (|a_i \rangle \langle a_j|_A \otimes |b_i \rangle \langle b_j|_B ) \rightarrow \left( \sum_{i,j=1}^{r_{\psi}} \lambda_i \lambda_j^* |a_i \rangle \langle a_j|_A   \right) \otimes \sigma_B \, .
    \label{eqn:decorrelation1}
\end{equation}

\begin{lemma}
    If the transformation presented in~\eqref{eqn:decorrelation1} is a unitary transformation, we have~$H(\sigma_B) \equiv 0$ with~$H(\cdot)$ being the von Neumann entropy. 
\end{lemma}
\begin{proof}
    Let us consider the von Neumann entropies:~$H(| \psi \rangle \langle \psi |_{AB}) = \sum_{i=1}^{r_{\psi}} |\lambda_i|^2 \log_2(|\lambda_i|^2)$ and $H\left( \sum_{i,j=1}^{r_{\psi}} \lambda_i \lambda_j^* |a_i \rangle \langle a_j|_A   \right) = \sum_{i=1}^{r_{\psi}} |\lambda_i|^2 \log_2(|\lambda_i|^2) = H(| \psi \rangle \langle \psi |)$. Suppose the transformation in~\eqref{eqn:decorrelation1} is done by a similarity transformation, i.e.,~$\rho \mapsto U \rho \, U^{-1}$ with~$U \in GL(\mathcal{H}_A)$ with~$\rho$ being an arbitrary state on~$S(\mathcal{H}_{AB})$, the total von Neumann entropy is constant during the transformation. Thus,~$H(\sigma_B)=0$ and then~$\sigma_B = | \Sigma \rangle \langle \Sigma |_B$ with~$| \Sigma \rangle_B $ being necessarily a pure state with von Neumann entropy zero. 
    
\end{proof}
\noindent
This relates to the preservation of local properties during the disentanglement process. 

\begin{equation}
    | \psi \rangle_{AB} = \sum_{i=1}^{r_{\psi}} \lambda_i (|a_i \rangle_A \otimes |b_i \rangle_B ) \rightarrow \left( \sum_{i=1}^{r_{\psi}} \lambda_i |a_i \rangle_A  \right) \otimes |\Sigma \rangle_B \, .
    \label{eqn:decorrelation}
\end{equation}
Here, system~$A$ denotes the quantum database, and system~$B$ is the sensor from which we read the data. During the writing process, both systems get entangled. For specific cases, e.g., if only classical correlations are present, techniques such as those presented in~\cite{cortese2018loading} can be used. 

\end{document}